\documentclass[11pt]{article}
\usepackage[hmargin=1in,vmargin=1in]{geometry}

\usepackage{amsmath,amssymb,amsfonts,amsthm}
\usepackage{libertinus}     
\usepackage{libertinust1math}
\usepackage[dvipsnames,usenames,table]{xcolor}
\usepackage{hyperref}
\hypersetup{colorlinks=true, urlcolor=Blue, citecolor=Blue!50!Cyan, linkcolor=BrickRed, breaklinks, unicode}

\usepackage{setspace}
\usepackage{bbm}
\usepackage{algorithm}
\usepackage[noend]{algpseudocode}
 \usepackage[framemethod=tikz]{mdframed}
\usepackage{xspace}
\usepackage{pgfplots}
\usepackage{bbm}
\usepackage{framed}
\usepackage{subcaption}
 \usepackage{thm-restate}
\usepackage{nicefrac}
 \pgfplotsset{compat=1.5}
 \usepackage[capitalise]{cleveref}
\usepackage{setspace}
\usepackage{mathtools}

\newtheorem{theorem}{Theorem}[section]

\newtheorem{lemma}[theorem]{Lemma}

\newtheorem{definition}[theorem]{Definition}
\newtheorem{remark}[theorem]{Remark}

\newtheorem{fact}[theorem]{Fact}

\newtheorem{question}[theorem]{Question}

\newcommand{\R}{\mathbb R}

\newcommand{\opt}{\text{OPT}}
\newcommand{\eps}{\varepsilon}

\newcommand{\poly}{\mathrm{ poly}}
\newcommand{\cost}{\mathrm{cost}}
\newcommand{\polylog}{\mathrm{ polylog}}
\renewcommand{\tilde}{\widetilde}
\DeclareMathOperator*{\argmin}{arg\,min}
\DeclareMathOperator*{\sat}{SAT}

\newcommand{\E}{\mathbb{E}}

\newcommand{\calA}{\mathcal A}
\newcommand{\calB}{\mathcal {B}}
\newcommand{\calC}{\mathcal {C}}
\newcommand{\calD}{\mathcal {D}}

\newcommand{\calP}{\mathcal {P}}
\newcommand{\calS}{\mathcal {S}}

\newcommand{\calL}{\mathcal{A}}

\newcommand{\badcut}{\eps}
\newcommand{\offset}{\tau(\eps, d)}
\DeclareMathOperator{\dist}{dist}
\newcommand{\globalS}{\opt}
\newcommand{\calH}{\mathcal{H}}
\newcommand{\bc}{\text{ b.c.}}
\DeclareMathOperator{\deto}{det}

\newcommand{\lpar}{\left(}
\newcommand{\rpar}{\right)}
\newcommand{\lbra}{\left[}
\newcommand{\rbra}{\right]}

\DeclareMathOperator{\level}{lvl}
\DeclareMathOperator{\diam}{diam}

\title{\textbf{Almost-Optimal Upper and Lower Bounds for\\ Clustering in Low Dimensional Euclidean Spaces}}
\author{}
\author{Vincent Cohen-Addad\\ Google Research\\ \texttt{\color{pink!70!red} cohenaddad@google.com}\and Karthik C.\ S.\footnote{This work was supported by the National Science Foundation under Grants CCF-2313372 and CCF-2443697, a grant from the Simons Foundation, Grant Number 825876, Awardee Thu D. Nguyen, and partially funded by the Ministry of Education and Science of Bulgaria's support for INSAIT, Sofia University ``St. Kliment Ohridski'' as part of the Bulgarian National Roadmap for Research Infrastructure.}\\ Rutgers University\\\vspace{0.1in} \texttt{\color{pink!70!red} karthik.cs@rutgers.edu}\and David Saulpic\\ Université Paris Cité, CNRS\\ \texttt{\color{pink!70!red} david.saulpic@irif.fr}\and Chris Schwiegelshohn\footnote{This work was partially supported by the Independent Research Fund Denmark (DFF) under a Sapere Aude Research Leader grant No 1051-00106B and by a Google Research Award.}\\ Aarhus University\\ \texttt{\color{pink!70!red} schwiegelshohn@cs.au.dk} }
\date{}

\begin{document}

\maketitle
\begin{abstract}
The $k$-median and $k$-means clustering objectives are classic objectives for
modeling clustering in a metric space. Given a set of points in a metric
space, the goal of the $k$-median (resp.\ $k$-means) problem is to find $k$ representative points so as to minimize the sum of the distances (resp.\ sum of squared distances) from each point to its closest representative.
Cohen-Addad, Feldmann, and Saulpic [JACM'21] showed how to obtain a $(1+\eps)$-factor approximation in low-dimensional Euclidean metric for
both the $k$-median and $k$-means problems in near-linear time $2^{(1/\eps)^{O(d^2)}} n \cdot \polylog(n)$ (where $d$ is the dimension and $n$ is the number of input points).\vspace{0.15cm}

We improve this running time to $2^{\tilde O(1/\eps)^{d-1}} \cdot n \cdot \polylog(n)$, and show an almost matching lower bound: under the Gap Exponential Time Hypothesis for 3-SAT, there is no $2^{{o}(1/\eps^{d-1})} n^{O(1)}$ algorithm
achieving a $(1+\eps)$-approximation for $k$-means.
\end{abstract}
\clearpage

\section{Introduction}
The $k$-median and $k$-means problems are minimization tasks of classic objectives for modeling
clustering in a wide variety of applications arising in data mining and
machine learning. Given a set of points $P\subset X$  and a set of \emph{candidate centers} $\calS \subset X$ in a metric space $(X, \Delta)$,
the goal of the $k$-median problem\footnote{This variant is sometimes referred to as the \emph{discrete} $k$-median problem. In the continuous $k$-median problem we are allowed to pick the $k$ centers $S$ from anywhere in $X$. } is to find a set of $k$ points $S \subset \calS$, called \emph{centers}, so as
to minimize the sum of the distances (given by $\Delta$) from each point of $P$ to its closest
center in $S$ (in the $k$-means problem, the goal is to instead minimize the sum of the squared distances from each point of $P$ to its closest center in $S$).

The algorithmic study of $k$-median started in the early '60s in the operations research community~\cite{kuehn1963heuristic,stollsteimer1961effect,hakimi1965optimum}, while the study of $k$-means was arguably
 initiated by the seminal work of Lloyd~\cite{lloyd1982least}, which applied clustering to quantization of analog signals. 
Since then, both problems have received a tremendous amount of attention.
Naturally, the complexity of the problems varies
with the underlying metric space hosting the input points. 
On the negative side, both problems are
known to be NP-hard in Euclidean space, even when the points lie in the Euclidean plane (and $k$ is large)~\cite{megiddo1984complexity, MahajanNV12}, or when $k=2$~\cite{dasgupta2009random,AKP24} (and the dimension is large).
When both $k$ and $d$ are arbitrary, 
several groups of researchers have shown hardness
of approximation results for both problems~\cite{AwasthiCKS15,LeeSW17,Cohen-AddadS19,Cohen-AddadSL22}. 
Even stronger hardness
results hold if the input lives in an arbitrary metric space (see Guha and Kuller~\cite{GuK99} and \cite{Cohen-AddadSL21}).

\paragraph*{Fine-Grained Complexity.}
Despite those hardness results, the importance of Euclidean inputs in statistics and
machine learning applications 
has led researchers to develop methods and algorithms, in particular through the study of the parameterized complexity of the problem, when the input lies in an Euclidean space. 
Both the dimensionality of the input, $d$, and the target number of clusters, $k$, have been studied as parameters. 
For exact algorithms, the seminal work of Inaba, Katoh and Imai~\cite{InabaKI94} has shown that one can compute an exact solution to the $k$-median and $k$-means problems in time $O(n^{kd+1})$ -- as opposed to the  naive brute-force enumeration which runs in $\tilde{O}(k^n)$ time.
On the negative side, Cohen-Addad, de Mesmay, Rotenberg and Roytman~\cite{CohenaddadSODA18} much later showed that there is no $n^{o(k)}$ exact algorithm for (discrete) $k$-median or $k$-means problem, even when
the dimension is four, and thus partially showing the optimality of the $O(n^{kd + 1})$ algorithm.

To circumvent this hardness, researchers have turned to approximation algorithms. When the dimension $d$ and the precision $\eps$ are taken as parameters, a line of work \cite{AroraRR98, KolliopoulosR07, FriggstadRS19, Cohen-AddadKM19, Cohen-Addad18, jacm} improved the running time of approximation schemes from quasi-polynomial time, $2^{(\log(n)/\eps)^{d-1}}$, to  $2^{(1/\eps)^{O(d^2)}} n$. 
The doubly exponential dependency in $d^2$ stands in stark contrast with other similar results, e.g., for the Traveling Salesperson Problem (TSP). 
For this problem, Kisfaludi-Bak, Nederlof and Wegrzycki~\cite{Kisfaludi-BakNW21} recently settled the exact dependency on $d$, by showing an algorithm running essentially in time $2^{O\lpar 1/\eps^{d-1}\rpar} n$, and a lower bound of $2^{\Omega\lpar (1/\eps)^{d-1}\rpar}$ under the Gap-ETH hypothesis \cite{D16,MR16} (see Definition~\ref{hyp:gap-eth}). 
In light of this result, we consider the following question:

\begin{question}\label{q:approx}
Is it possible to obtain a $2^{O\lpar 1/\eps^{d-1}\rpar} n^{O(1)}$ approximation scheme for $k$-median and $k$-means problems in low-dimensional inputs? Would it be possible to get an even faster $2^{o\lpar 1/\eps^{d-1}\rpar} n^{O(1)}$ approximation scheme for both problems?
\end{question}

\subsection{Our Results}
In this paper, we make substantial progress towards answering \cref{q:approx}. Our first contribution is the design of an improved algorithm:

\begin{theorem}\label{thm:mainAlgo}
    For every $\eps > 0$ and dimension $d$, the $k$-median and $k$-means problems in $\R^d$ can both be  approximated to a $(1+\eps)$-factor in time $2^{\tilde O\lpar 1/\eps^{d-1}\rpar} n \cdot \polylog(n)$.\footnote{The $\tilde O$ notation hides an exponential dependency in $d$, and polynomial dependency in $\log(1/\eps)$.}
\end{theorem}

The above algorithm is our main contribution, and improves on the techniques of \cite{jacm}. It also readily extends to the continuous $k$-median and $k$-means problems (see \cref{rem:cont}). 

On the technical side, this algorithm deepens our understanding of the \emph{quadtree} decomposition for clustering.
A long line of work based on this decomposition \cite{AroraRR98, KolliopoulosR07, FriggstadRS19, Cohen-AddadKM19, Cohen-Addad18, jacm} improved the running time of approximation schemes (i.e., $(1+\eps)$-factor approximation
algorithms for any small enough constant $\eps$) for clustering from quasi-polynomial time, $2^{(\log(n)/\eps)^{d-1}}$, to  $2^{(1/\eps)^{O(d^2)}} n$ \cite{jacm}. 
The main conceptual insight of these works is to bound the number of \emph{portals} necessary along the decomposition so that forcing a path to go through portals incurs at most a $(1+\eps)$-factor distortion.
The main contribution of \cite{jacm} is that after processing a tiny fraction of the input, only a constant number of portals are sufficient at each level of the quadtree decomposition.
However, this constant remains much larger than that required for other natural problems.

Indeed, carefully analyzing the performances of the quadtree dissection and the number of required portals to achieve a $(1+\eps)$-factor approximation has been a fruitful direction of research to design approximation schemes in Euclidean space.
The initial work of Arora \cite{Arora96, Arora98} showed approximation schemes for the Traveling Salesperson Problem (TSP) and other routing problems, first in quasi-polynomial time (for $d > 2$) and then with running time $n \log(n)^{O(1/\eps)^{d-1}}$. 
For TSP, the running time was subsequently significantly improved: first by Rao and Smith~\cite{rao1998approximating} who achieved $(1/\eps)^{1/\eps^{d-1}} n \log n$, then by Bartal and Gottlieb \cite{BartalG13} who brought it down to $2^{1/\eps^{O(d)}} n$, linear in $n$, but suboptimal in $d$. 
Finally, \cite{Kisfaludi-BakNW21} settled the exact dependency on $d$, by showing an algorithm running essentially in time $2^{O\lpar 1/\eps^{d-1}\rpar} n$, and a lower bound of $2^{\Omega\lpar (1/\eps)^{d-1}\rpar} \poly(n)$ under the Gap-ETH hypothesis \cite{D16,MR16} (see Definition~\ref{hyp:gap-eth}).

In light of this literature, our contribution is therefore an improved analysis of the quadtree dissection for clustering problems, that allows us to (almost) match the number of portals required for the well-studied TSP. This requires an analysis and techniques distinct from those for TSP, in particular because they fail if distances are squared -- as what happens in the $k$-means problem.
We note that, due to its simplicity, applications of the quadtree extend far beyond approximation algorithms in low-dimensional space. 
For clustering, quadtrees are used for instance in the streaming scenario, when memory has to be kept low \cite{BravermanFLSY17}, but also to design linear time algorithms in high-dimensional spaces \cite{rejectionSampling} or to construct differentially-private algorithms~\cite{kdd}. 
Therefore, better understanding the quadtree's performance, even in the simplest scenario, is likely to have a far-reaching impact.

We complement Theorem~\ref{thm:mainAlgo} with a conditional lower bound, which shows the near-optimality of our algorithm under the Gap-ETH hypothesis:

\begin{restatable}{theorem}{infLD}
    \label{thm:infLD}
Assuming the Gap-ETH hypothesis, for every integer $d \ge 2$, there exists a constant $c > 0$ such that no approximation scheme can, given an instance of discrete $k$-means (or discrete $k$-median) with $N$ points in Euclidean space $\mathbb{R}^d$ and a parameter $\eps > 0$, compute a $(1+\eps)$-factor approximation in time $2^{c(1/\eps)^{d-1}} \poly(N)$.
\end{restatable}

This lower bound is obtained by amalgamating the technical ideas of \cite{LeeSW17,CohenaddadSODA18,Cohen-AddadS19} and the framework of de Berg, Bodlaender, Kisfaludi-Bak, Marx and van der Zanden \cite{BergBKMZ20}. Ideas to reduce the Vertex Cover problem to clustering problems were elaborated in \cite{LeeSW17,CohenaddadSODA18,Cohen-AddadS19}, and \cite{BergBKMZ20} built a framework  to prove fine-grained lower bounds for a host of geometric problems in low dimensions. We put them together to obtain Theorem~\ref{thm:infLD}.

\paragraph*{Extensions.}
The framework developed in \cite{jacm} applies to many variants of $k$-means (and $k$-median), namely prize-collecting $k$-means, $k$-means with outliers and Facility location. As our structure theorem is a direct improvement of theirs, we also improve the running time for those problems. We give a sketch of the arguments in Appendix \ref{app:extension}. 
Our analysis also improves the result for doubling metric (a generalization of Euclidean space), however attaining a non-tight complexity $2^{\tilde O(1/\eps^d)}$. 
We therefore focused this paper on the Euclidean setting, for which our analysis is almost tight.

\subsection{Our Techniques}\label{sec:overview}
For the ease of presentation, we detail below our upper bound for the $k$-means problem in the Euclidean plane (i.e., with $d=2$). Our input is a set of points  $P \subset \R^2$, a set of candidate centers $\calC\subset \R^2$, and the goal is to compute a set of centers $S \in \calC^k$ ($\calC^k$ is the $k$-wise partition product of $\calC$) that minimizes the sum of  squared distances from points of $P$ to their closest center in $S$.

\paragraph*{Upper bound.} Our starting point is the quadtree dissection equipped with portals. We briefly recall in the next few  paragraphs its construction and main property. 

First, find an axis-aligned rectangle that contains all input points. Then, randomly split the rectangle into 4 axis-aligned rectangles of side length reduced by a factor of approximately 2. Continue this splitting recursively until each rectangle contains a single point: this defines a recursive decomposition of the input. 
We say that two points are \textit{cut at level $i$} when the smallest rectangle containing both points is at the $i^{\text{th}}$ level of the decomposition. Assuming that the leaves of the decomposition have diameter 1, the rectangles at the $i^{\text{th}}$ level have diameter roughly $2^i$.

Now, a set of \emph{portals} is placed regularly along the boundary of each rectangle. 
Instead of connecting points to centers via straight lines, we consider only \emph{portal-respecting} paths, defined as follows. To connect two points $p$ and $q$ cut at level $i$, a portal-respecting path consists of a sequence of segments connecting $p, p_1, ..., p_i, q_i, ..., q_1, q$, where $p_j$ (resp.\ $q_j$) is a portal of the rectangle at level $j$ containing $p$ (resp. $q$). 
Considering only those paths enforces some detours, but a simple dynamic program can compute the best portal-respecting solution, whose complexity naturally depends (exponentially) on the number of portals. 
Furthermore, if $1/\rho$ many portals are placed along each boundary, it can be computed that a portal-respecting path between two points $p, q$ cut at level $i$ has length (essentially) $\|p-q\| + \rho 2^i$. Indeed, since the boundaries of the rectangles are  geometrically decreasing, the detour to go through a portal at level $i$ dominates all others; and since the rectangles at this level have diameter $2^i$, by regularly placing    $1/\rho$ many portals, we can ensure that there is a portal at distance at most $\rho 2^i$ from the point where the straight line path crosses the boundary.
Therefore, the goal of the analysis is to show that, for a small number of portals, a portal-respecting solution with almost optimal cost exists. 

The standard argument goes as follows. The main property of the quadtree decomposition is that, for any pair of points $p, q \in \R^2$, the probability that they are cut at level $i$ is $\|p-q\| / 2^i$. By using $1/\rho$ portals, the detour to connect $p$ and $q$ via portals in case they are cut at level $i$ is $\rho 2^i$.
Therefore, the expected detour at level $i$ is $\|p-q\| / 2^i \cdot \rho 2^i = \rho \|p-q\|$. Since there are $\log n$ levels in the decomposition, the expected detour is $\rho \|p-q\| \cdot \log n$, and one has to take $\rho = \eps / \log n$ to get a tiny expected detour.

However, having $\log n / \eps$ many portals does not yield an efficient algorithm. 
More importantly, this analysis does not extend to $k$-means problem: although the expected distance between $p$ and $q$ is small, the expected \emph{squared} distance is not! Indeed, the expected squared detour at level $i$ is $\|p-q\| / 2^i \cdot (\rho 2^{i})^2 = \rho^2 \|p-q\| 2^i$, and the term $2^i$ cannot be easily controlled.
To cope with this, \cite{jacm} introduced a careful preprocessing of the input, to ensure that no point is cut at a level higher than $\|p-q\|/\eps$. In that case, taking $1/\eps^2$ portals is enough to ensure a small detour -- and also small squared detour, as this is not an average-case argument (unlike the previously seen argument). This preprocessing is essentially the following: start from a constant-factor approximate solution $\calA$ (which is a set of $k$ centers), and for each point $p$, let $\calA_p$ be its distance to the closest center of $\calA$. 
If the ball $B(p, \calA_p/\eps)$ is cut by the decomposition at a level higher than $\log(\calA_p/\eps^2)$, then replace $p$ by a copy of its closest center in $\calA$. They showed that each point is replaced with probability at most $\eps$, therefore yielding an expected cost increase of $\eps \cost(P, \calA) = O(\eps)\opt$. 
Furthermore, \cite{jacm} showed that placing $1/\eps^3$ many portals ensured that,  for any center of $\calA$, either it can be connected to the optimal solution through portals with a small detour, or one can assume that the center is part of the optimal solution. Therefore, after preprocessing, each client can be connected through portals to the optimal solution, and dynamic programming allows to compute a near-optimal solution.

However, this analysis may be seen as ``worst-case", in the sense that it treats  all the clients for which the ball is cut at a level below $\log(\calA_p/\eps^2)$ in the same way. 
Our contribution is to mix the average-case analysis with the techniques of \cite{jacm} in order to reduce the number of necessary portals. 
We carefully define a budget for each point, according to the level where the point is cut from the approximate solution $\calA$, as in \cite{jacm}, but now also according to the level where it is cut from the optimal solution. 
We first show that, with constant probability, this budget is very cheap, namely an $\eps$-fraction of the optimal cost. 
We then show that this budget is enough to pay for the preprocessing of the input, and for the detour incurred by making a solution portal-respecting.  The algorithm is the same as the one of \cite{jacm}: our contribution lies in a much tighter analysis.

This analysis differs from \cite{jacm} as the detour they tolerated  was independent of the optimal solution: using this additional information allows us to be more precise, and reduce the number of required portals to $(\log(1/\eps)/\eps)^{d-1}$, instead of $1/\eps^{O(d)}$. 
Although the global picture remains similar, finding the right balance between a budget (1) small enough to not blow up the cost, and (2) that is still enough to pay for making the optimal solution portal-respecting is a non-trivial extension of their proof. 
Indeed, it requires new insights on the way to connect a point $p$ that is \emph{badly cut}, namely cut from its optimal center $s$ at a level higher than $\log\lpar \|p-s\|/\eps\rpar$.
For this, we show that our budget allows to connect the point in a portal-respecting way to the center $\tilde s$ closest to $\calL(p)$, where $\calL(p)$ is the closest center to $p$ in the constant-factor approximate solution $\calL$.

\paragraph*{Lower bound.}
Our lower bound construction builds upon the techniques of de Berg et al.~\cite{de2018framework}. Given a 3-SAT formula, they show how to construct a graph embedded into $\R^d$ such that solving Vertex Cover on the embedded graph is equivalent to deciding whether the formula has a valid assignment or not. This graph has $|E| = n^{d/(d-1)}$ many edges. We use the (by now standard, see e.g.~\cite{CohenaddadSODA18}) relationship between $k$-means and Vertex Cover, and build a $k$-means instance so that computing its optimal cost allows to determine whether the formula is satisfiable or not. 

To cope with approximation, we observe that the framework of \cite{de2018framework} can be readily extended to handle constant gaps in the completeness and soundness cases by relying on Gap-ETH \cite{D16,MR16} instead of ETH (much like in the work of \cite{Kisfaludi-BakNW21}). Thus, we show that any $(1+\eps)$-factor approximation to $k$-means leads to a vertex cover that covers almost all the edges, leaving behind only $O(\eps |E|)$ many edges uncovered.

Thus, starting from a hardness of approximation instance    of the Vertex Cover problem which can be embedded on the grid in $\mathbb{R}^d$, where every edge can be thought of as a line-segment, we think of the mid point of these line segments as the points to cluster and the points corresponding to the vertices as the set of candidate centers. A crucial property of this clustering instance embedding is that the only candidate centers close to the midpoints of an edge are the points corresponding to the two vertices that it is incident on. In the completeness case, by picking the points corresponding to the vertex cover we have that every point is close to one of the $k$-centers (in particular the vertex that covers that edge), and in the soundness case one can recover from any $k$ centers with a low clustering cost, a vertex cover for the original graph that covers most edges.

\subsection{Further Related Work}\label{app:relatedWork}

We covered previously the hardness results for $k$-means and $k$-medians; in terms of upper-bound, for general metrics a recent breakthrough from \cite{Cohen-Addad0LSS25} achieves a $(2+\eps)$-approximation for $k$-median, and a $5.6$-approximation for $k$-means \cite{kmeans}. Those two works improve a long line of works, based on primal-dual techniques and the so-called Lagrangian-multiplier preserving algorithms for Facility Location, see e.g. \cite{JainV01, JMS02, LiS16}. 

To improve those approximation ratio, one can resort to algorithms parameterized by the number of clusters, $k$: one can get FPT algorithm with approximation matching the lower bound of \cite{JMS02}, namely $1+2/e$ for $k$-median and $1+8/e$ for $k$-means \cite{Cohen-AddadG0LL19}. 

Other techniques, specific to Euclidean Space, have been developed in the past decades. Dimension reduction \cite{MakarychevMR19} and coreset \cite{Cohen-AddadSS21} allow to reduce the dimension of the input to $O(\log k / \eps^{-2})$ and the number of distinct points to $\tilde{O}(k \eps^{-O(1)})$, hence leading to simple FPT algorithms based on naive enumeration.
Recent papers improved this simple approach, reducing the dependency in $k$ (in the exponent) \cite{AbbasiBBCGKMSS23}. 

We covered in the introduction most of the literature related on low-dimensional clustering -- which is the study of the problems parameterized by $d$. To compute a $(1+\eps)$-approximation in time polynomial in $n$ and $k$, any algorithm must have a running time at least doubly exponential in $d$, as the problem is APX-hard in dimension $\Omega(\log n)$. The best of these algorithms is from \cite{jacm}, with a near-linear running time of $f(\eps, d)\cdot  n \operatorname{polylog} n$. 

Finally, tractability of clustering has also been studied with different parameters, e.g., the cost by Fomin, Golovach and Simonov \cite{FominGS21}, who showed a $D^D \mathrm{poly}(nd)$ exact algorithm for $k$-median, where $D$ is the optimal cost. 

\subsection{Organization of the paper}
We focus the presentation on the case of $k$-means, as dealing with squared distances with quadtree is arguably harder. In Section~\ref{sec:prelim}, we describe the necessary notions and definitions to prove Theorem~\ref{thm:mainAlgo}. In Section~\ref{sec:alg} we analyze our PTAS yielding Theorem~\ref{thm:mainAlgo} and then in Section~\ref{sec:infLD} we show the lower-bound result obtaining Theorem~\ref{thm:infLD}.

We present in Appendix~\ref{app:k-med} how to adapt the argument for $k$-median. 
The main body contains the proof of the structure theorem that allows to reduce the number of portals required by the quadtree. 
We  extend the algorithmic results to other objectives in Appendix \ref{app:extension}.

\section{Preliminaries}\label{sec:prelim}
\subsection{Definitions}

We consider the Euclidean space $\R^d$, with the $\ell_2$ metric: $\dist(p, q) := \sqrt{\sum_{i=1}^d (p_i - q_i)^2}$. 
For any point $p$ and $r \geq 0$, the ball centered at $p$ of radius $r$ is $B(p,r) := \{x \in \R^d: \dist(p, x) \leq r\}$.
The (discrete) $k$-median and $k$-means problems are defined as follows: we are given a set of \emph{clients} $P$ and of \emph{candidate centers} $\calC$. A \emph{solution} is any subset of $\calC$ with size $k$. 
The goal is to compute the solution $\calS$ that minimizes $\cost(P, \calS) := \sum_{p \in P} \dist(p, \calS)^z$, with $z=1$ for $k$-median and $z=2$ for $k$-means.
We say a solution is an $\alpha$-approximation if its cost is at most $\alpha$ times the minimal cost.

\begin{remark}\label{rem:cont}
    We consider in this paper the \emph{discrete} Euclidean clustering. In the continuous version, centers can be placed arbitrarily in the space. A classical result shows how to reduce continuous to discrete, by computing a set of $1/\eps^d \cdot |P|$ candidate centers that contain a $(1+\eps)$-approximation \cite{Matousek00}.
\end{remark}

For some set of points $\calS$, and any point $p \in \R^d$, we define $\calS(p) = \argmin_{s \in \calS} \dist(p, s)$ be the closest point in $\calS$ to $p$ ; and we let $\calS_p :=\dist(p, \calS(p) )$ be the distance from $p$ to its closest point in $\calS$.

We assume that the minimal distance between two points is $1$, and denote $\diam(P)$ the maximum distance between two points of $P$. We also assume that $P \subset B(0, \diam(P))$. Using standard arguments for clustering problems, we can assume $\diam(P) = \poly(n)$ \cite{CohenAddadFS19}.

A hierarchical decomposition $\calD$ is a sequence $\calB_1, \calB_2 ...$ such that for any $i$, $\calB_i$ is a partition of $P$, and $\calB_i$ is a refinement of $\calB_{i+1}$: namely, every part of $\calB_i$ is fully contained in a part of $\calB_{i+1}$ (note that in this definition, parts are increasing with $i$).

\subsection{The Hierarchical Decomposition and Its Properties}

\begin{lemma}\label{lem:talwar-decomp}
For any set of points $P \subset \R^d$ and any $\rho > 0$, 
there is a randomized hierarchical 
decomposition $\calD$ such that the diameter of a part $B \in \calB_i$ is at most
$2^{i+1}$, $|\calD|\leq\lceil\log(\diam(P))\rceil$, each part $B \in \calB_i$ is refined in at most $2^{O(d)}$ parts at level $i-1$, and:

\begin{enumerate}

  \item\label{prop:doub:prob} for any $p \in \R^d$, 
radius $r$, and level $i$, we have 
    \[\Pr[\calD \text{ cuts } B(p,r) \text{ at a level } i] \leq  dr/2^i.\]

 \item\label{prop:doub:portals} for every part $B\in \calB_i$ of the decomposition, there is a set of at most $ 1/\rho^{d-1}$ many 
    \emph{portals} such that for every points $p \in B, q \notin B$ there is a portal $x\in\calP_B$ close to the line connecting $p$ to $q$, i.e., 
$\|p-x\| + \|x-q\| \leq \|p-q\| + \rho 2^{i+1}$.
Furthermore, any portal of level $i+1$ that lies in $B$ is also a 
portal of $B$.
    \end{enumerate}
This decomposition can be found in time $(1/\rho)^{O(d)} n \log (\diam(P))$.
\end{lemma}

In Euclidean space, this decomposition is precisely the quadtree dissection. The previous statement is an adaptation of Lemma 2.3 in~\cite{jacm} to Euclidean space. The improved Property~\ref{prop:doub:prob} is shown, e.g., in Lemma 11.3 of \cite{har2011geometric}.
Since $\diam(P) = \poly(n)$, the decomposition can be computed in near-linear time.

We denote 
$\offset = \log(d) + \log(1/\badcut)$,
a parameter often used throughout this paper.

A \emph{portal-respecting} path between two points $p$ and $q$ cut at level $i$ consists of a sequence of segments connecting $p, p_1, ..., p_i, q_i, ..., q_1, q$, where $p_j$ (resp. $q_j$) is a portal of the rectangle at level $j$ containing $p$ (resp. $q$). 
By extension, a solution to $k$-means is portal-respecting if all paths connecting clients to center are portal-respecting. Recall that $\calL_p$ is the distance from $p$ to the closest point in a solution $\calL$. An important definition we will use is that of badly cut, slightly modified from \cite{jacm}:

\begin{definition}\label{def:badlycut}
  Let $P\subset \R^d$, $\calD$ be 
a hierarchical decomposition of $P$, and $\eps > 0$. 

For any ball $B(x, r)$, we define $\level(x, r)$ the level at which the ball $B(x,r)$ is cut by $\calD$. 

We say a ball $B(x, r)$ is \emph{badly cut w.r.t.\ $\calD$} if 
$\level(x, r) \geq \log(3r) + \offset$.
We say a point $p$ is \emph{badly cut w.r.t.\ $\calD$ and $\calL$} if $B(p, 3\calL_p)$ is badly cut w.r.t.\ $\calD$ (namely,
$\level(p, 3\calL_p) \geq \log(3\calL_p) + \offset$).
\end{definition}

In the following, we will consider badly cut points w.r.t.\ two different sets: 
$\calL$ will be a  constant factor approximation, known to the algorithm, and $\calS$ a slight modification of the optimal solution, unknown to the algorithm, that we will specify later; those two solutions being independent of the decomposition $\calD$.
For input points in $P$, we will examine whether they are badly cut w.r.t.\ $\calD$ and $\calL$. 
For centers of $\calL$, it will be w.r.t.\ $\calD$ and $\calS$. 
Furthermore, $\calD$ will always be the same, hence we  say for simplicity badly cut w.r.t.\ $\calL$ (or w.r.t.\ $\calS$).

The following lemma bounds the probability of being badly cut, and is a direct consequence of Property~\ref{prop:doub:prob}.
\begin{lemma}\label{lem:badlycutddim}
  Let $P, S \subset \R^d$ and $\calD$ a random hierarchical 
  decomposition given by \cref{lem:talwar-decomp}, with $\rho = O(\eps/\log(1/\eps))$.
 For any point $p \in \R^d$, the probability that $p$ is badly cut w.r.t.\  $S$
  is at most $\badcut$.
\end{lemma}

Given a random decomposition, each point is assigned a \emph{budget}. In the following definition, one can think of $\calL$ as an $O(1)$-approximation known by the algorithm, and $\calS$ as the optimal solution. We show in \cref{lem:budgetBounded} that this budget is small, and that one can compute a solution with cost at most the optimal cost plus the budget.

\begin{definition}\label{def:budget}
Let $P\subset \R^d$, $\calD$ be 
a hierarchical decomposition of $P$, $\eps > 0$, $\calL$ and $\calS$ be two 
solutions to $k$-means on $P$. 

For a ball $B(x, r)$, its \emph{detour } with respect to $B$ and $\calD$ is $\deto_\calD(x, r) = \eps 2^{\level(x, r)} r + \eps^2 2^{2\level(x, r)}$.

A point $p\in P$ has \emph{budget} with respect to $\calL$ and $\calS$ $b(p, \eps) = b_1(p, \eps) + b_2(p, \eps) + b_3(p, \eps)$, with
\begin{align*}
    b_1(p, \eps) &= \begin{cases} 
 \deto_\calD(p, 3\calL_p)\text{ if } \level(p, 3\calL_p)\leq \log(3\calL_p) + \offset \\
  0\text{ otherwise }
\end{cases}
\\
b_2(p, \eps) &= \begin{cases} 
\deto_\calD\lpar p, 3(\calL_p + \calS_p)\rpar
 \text{ if } \level(p, 3(\calL_p + \calS_p))\leq \log(3(\calL_p + \calS_p)) + \offset \\
 36 d  \calL_p^2 +  16  d \calS_p^2 \text{ otherwise }
\end{cases}
\\
b_3(p, \eps) &= \begin{cases} 
\deto_\calD(\calL(p),  3\calS_{\calL(p)})
 \text{ if } \level(\calL(p), 3\calS_{\calL(p)})\leq \log(3\calS_{\calL(p)}) + \offset \\
 0 \text{ otherwise }
\end{cases}
\end{align*}
\end{definition}

\section{Structure Theorem and the Approximation    Algorithm}\label{sec:alg}

\subsection{Roadmap}
Now that all definitions are in place, we give a more precise high-level overview of the proof of \cref{thm:mainAlgo}. 
Our goal is to show that the decomposition of \cref{lem:talwar-decomp} with $\rho = \eps/\log(1/\eps)$ is precise enough for a dynamic program to compute a portal-respecting solution that is a $(1+\eps)$-approximation.

To do this, we show that one can find a solution $\calS^*$ such that: (1) it is a $(1+\eps)$-approximation, and (2) the detour from each point to $\calS^*$ is at most the budget of the point. To conclude, it is enough to bound the budget, which is a mere application of \cref{lem:talwar-decomp} and first-moment analysis.
To find such a solution $\calS^*$, we start from the optimal solution (which respects (1) but not (2)), and transform it such that each client can be connected through the portals. 

First, consider the easy case where a client $p$ is not badly cut, and $B(p, 3(\calA_p + \opt_p))$ is not badly cut. Then, the detour from $p$ to $\opt_p$ can be charged to $b_2(p, \eps)$, as the closest optimal center to $p$ lies in the latter ball. 
When this ball is badly cut, instead of serving $p$ by its closest center $\opt(p)$, we serve it by $\opt(\calL(p))$. We show that this can always be charged to the budget of $p$. Indeed, in the sub-case where $\calL(p)$ itself is not badly cut, we can show that the detour between $p$ and $\opt(\calL(p))$ is affordable compared to the distance between them; and this distance can be charged to the budget $b_2(p, \eps)$. This sub-case is illustrated in \cref{fig:caseBad}. 
However, in the sub-case where $\calL(p)$ is badly cut as well, we have no choice but to add the center to $\opt$. This results in a solution with slightly too many centers -- in expectation $(1+\eps)k$, since each center is badly cut with probability $\eps$ -- and we will need to remove some centers of the optimal solution without increasing the cost too much. This step is already presented in \cite{jacm}, and we recall it in \cref{sec:prepareInstance}.
\begin{figure}
    \centering
    \includegraphics[scale=0.5]{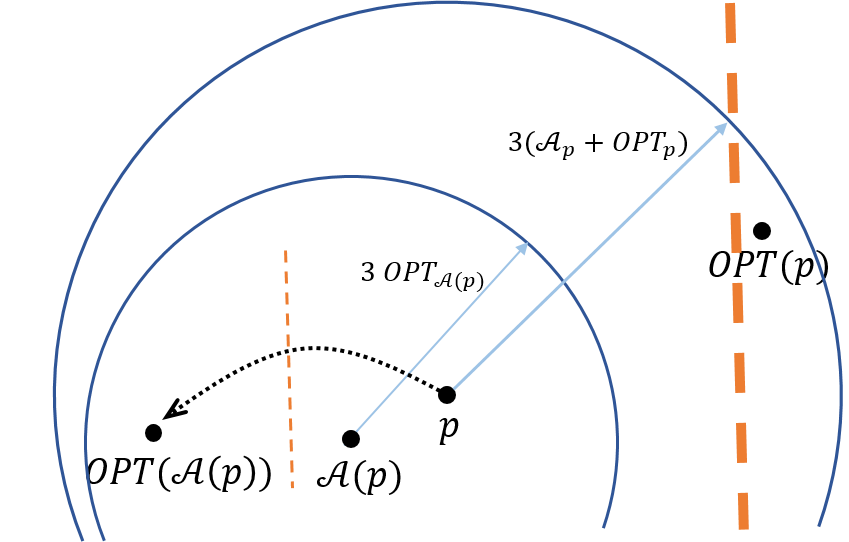}
    \caption{Illustration of one case of the distinction. $B(p, 3(\calA_p + \opt_p))$ is badly cut by the thick orange-dashed line, so $p$ cannot be connected via portals to $\opt(p)$. However, $p$ and $\calA(p)$ are not badly cut, so $p$ can be connected to $\opt(\calA(p))$ instead, making a detour through the thin orange-dashed line. The cost of this reassignment is charged to $b_2(p)$.}
    \label{fig:caseBad}
\end{figure}

This essentially concludes the case where $p$ is not badly cut. When it is badly cut, the decomposition $\calD$ approximates its distances poorly; however, this happens with tiny probability, and we can afford to move $p$ to $\calA(p)$ (charging this movement to $b_1(p, \eps)$). A case analysis, similar to the previous one, then shows that the detour between $\calA(p)$ to $\opt$ can be charged to $b_2(p, \eps)$ and $b_3(p, \eps)$.

In the remainder of this section, we show how to remove some centers from the optimal solution (to then add the badly cut centers of $\calA$), then prove the structure theorem that shows the existence of a solution satisfying (1) and (2). 
We explain in Section~\ref{sec:algo} the dynamic program that relies on this structure theorem and enables the computation of a good approximation, to conclude the proof of \cref{thm:mainAlgo}.

\subsection{Preparing the Instance}\label{sec:prepareInstance}
We consider the mapping of the
centers of $\globalS$ to $\calL$ defined as follows: for 
any $f \in \globalS$, let $\calL(f)$  denote the center of $\calL$ that is the closest 
to $f$. Recall that for a client~$c$, $\calL(c)$ is the center serving $c$ in $\calL$.

For any center $\ell$ of~$\calL$, define $\psi(\ell)$ to be the set of centers 
of $\globalS$ that are mapped to $\ell$, namely, $\psi(\ell) = \{f \in \opt \mid 
\calL(f) = \ell \}$. 
Define $ \calL^1$ to be the set of centers $\ell$ of $\calL$ for 
which there exists a unique $f \in \globalS$ such that $\calL(f) = \ell$, namely 
$ \calL^1 = \{\ell\in \calL \mid |\psi(\ell)| = 1\}$. Let $ \calL^0 = \{ \ell\in \calL \mid 
|\psi(\ell)| = 
0\}$, and $ \calL^{\geq 2} = \calL \setminus( \calL^1 \cup \calL^0)$. 
Similarly, define 
$\globalS^1 = \{f \in \globalS \mid \calL(f) \in \calL^1\}$ 
 and $\globalS^{\geq 2} = \{ f \in \globalS \mid 
\calL(f) \in \calL^{\geq 2}\}$. Note that $|\globalS^{\geq 2}| = | \calL^0|+| \calL^{\geq 2}|$, since 
$|\globalS^1|=| \calL^1|$ and, w.l.o.g., $|\globalS|=| \calL|=k$.

As shown in \cite{jacm}, one can remove some of the centers from $\globalS^{\geq 2}$ to obtain the following result:
\begin{lemma}[See Step 1. and Lemma 4.1 in \cite{jacm}]\label{lem:coststep1}
There exists a set $\calH \subseteq \globalS^{\geq 2}$ such that: 
\begin{itemize}
    \item $|\calH| = \lfloor \eps\cdot  |\globalS^{\geq 2}| / 2\rfloor$,
    \item for any $\ell \in \calL$, the center in $\psi(\ell)$ closest to $\ell$ is not in $\calH$, 
    \item let $\opt' = \opt \setminus \calH$: $\opt'$ has cost at most $(1+O(\eps))\cost(\globalS) + O(\eps)\cost( \calL)$.
\end{itemize}
\end{lemma}

Furthermore, we transform the input set $P$ to give it more structure, as follows.
Given a randomized hierarchical decomposition $\calD$ and a solution $\calL$, we define the set $\tilde P$ as follows. For every point $p$, 
\begin{equation*}
\tilde p = \begin{cases}
    \calL(p) \text{ if $p$ is badly cut w.r.t.\ } \calD \text{ and } \calL\\ p \text{ otherwise.}
\end{cases}
\end{equation*}
We note that $\tilde P := \{\tilde p: p \in P\}$ is a random variable that depends on the randomness of~$\calD$.

\subsection{Construction of a Structured Solution}\label{sec:constructS*}

Using \cref{lem:coststep1}, we can now show the existence of a near-optimal portal-respecting solution. 
Let $\calL$ be a fixed solution, and $\opt'$ the solution from \cref{lem:coststep1} for $P$.
 Let $\calD$ be a randomized hierarchical decomposition, and $B_{\calD}$ be the set of centers of $\calL$ that are badly cut w.r.t.\ $\opt'$, i.e., $c\in B_{\calD}$ when the ball $B(c,3\opt'_c)$ is cut at some level greater than $\log(3\opt'_c)+\offset$.

\begin{definition}\label{def:smallDist}
For $\eps  > 0$, we say that $\calD$ has $\eps$-\emph{small distortion} if 
\begin{enumerate}
\item \label{prop:budget} 
the total budget w.r.t.\ $\calL$ and $\opt'$ is bounded: $\sum_{p \in P} b(p, \eps) \leq \eps (\cost(P, \opt) + \cost(P, \calL))$.
\item\label{prop:costSstar} there exists a solution $ \calS^*$ that contains $B_{\calD}$ 
with $\sum_{p \in P} \lpar \dist(p, \tilde p) + \dist(\tilde p,  \calS^*)\rpar^2 \leq (1+\eps)\cost(P, \opt) + \eps \cost(P, \calL)$,
\item\label{prop:det} for each point $p \in P$, its detour to $\calS^*$ is at most $b(p, \eps)$, namely there is a center $s \in \calS^*$ cut from $\tilde p$ at level $i$ such that $(\dist(p, \tilde p) + \dist(\tilde p, s) + \eps 2^i)^2 \leq \lpar\dist(p, \tilde p) + \dist(\tilde p, \calS^*)\rpar^2 + b(p, \eps)$.
\end{enumerate}
\end{definition}

As explained in the introduction, we will use those properties as follows: a standard dynamic program will compute the best \emph{portal-respecting} solution that minimizes $\tilde{\cost}(P, \calS) := \sum_p (\dist(p, \tilde p) + \dist(\tilde p, \calS))^2$. Making the solution portal-respecting incurs an error $\eps 2^i$ for two points cut at level $i$: Property \ref{prop:det} therefore bounds this error by $b(p, \eps)$. Property \ref{prop:budget} of \cref{def:smallDist} ensures that the total budget is bounded, and Property \ref{prop:costSstar} relates $\tilde{ \cost}$ to $\cost$.

Our main structural theorem for $k$-means in low dimension is the following:

\begin{theorem}[Structure Theorem]\label{thm:smallDist}
    For any $\eps > 0$, $\calD$ has $\eps$-small-distortion with probability at least $2/3$.
\end{theorem}
The rest of the paper is dedicated to the proof of this theorem.

\subsubsection{Bounding the Budget}
\begin{lemma}\label{lem:budgetBounded}
Fix two solutions $\calL$ and $\calS$, and let $b(p, \eps)$ be the budget of $p$ w.r.t.\ solutions $\calL$ and $\calS$. With probability  $2/3$ (over the randomness of $\calD$),   $\sum_p b(p, \eps) = O(d^2 \log(1/\eps) \cdot \eps) (\cost(P, \calL) + \cost(P, \calS))$
\end{lemma}
\begin{proof}
 We bound the expectation of the budget and then use Markov's inequality. Recall that any ball $B(x, r)$ is cut at level $i$ with probability $d r / 2^i$, and is  badly cut with probability $\badcut$. Furthermore, when $B(x, r)$  is cut at level $i$, then $\deto_\calD(x, r) = \eps r 2^i + \eps^2 2^{2i}$. 
 
 To simplify the analysis slightly, we first analyze the expectation of the variable 
\[\beta(x, r) := \begin{cases}\deto_\calD(x, r) \text{ when } \level(x, r) \leq \log(r) + \offset\\
 0 \text{ otherwise.}\end{cases}\]
 To do so, we distinguish between the cases where $\level(x, r)$ is at most $\log(r)$, between $\log(r)$ and $\log(r) + \offset$, and at least $\log(r) + \offset$. This gives the following upper bound:
 \begin{align*}
     \E[\beta(x, r))]
     &\leq \sum_{i \leq \log(r)} \deto_\calD(x, r) + \sum_{i = \log(r)}^{\log(r) + \offset} \deto_\calD(x, r) \cdot \Pr[\calD \text{ cuts } B(p,r) \text{ at a level } i]\\
     &\leq \sum_{i \leq \log(r)} \lpar \eps  2^i r + \eps^2 2^{2i} \rpar + \sum_{i = \log(r)}^{\log(r) + \offset} \lpar  \eps 2^i r + \eps^2 2^{2i}\rpar \cdot \frac{d r}{2^i}\\
     &\leq 3\eps r^2 + \sum_{i = \log(r)}^{\log(r) + \offset} \lpar \eps d r^2 +  \eps^2 2^{i}  \cdot d r\rpar\\
     &\leq 3\eps r^2 + \offset \cdot \eps d r^2 +  \eps^2  \cdot d r \cdot 2^{\log(r) + \offset}\\
     &= O(d^2\log(1/\eps)) \eps r^2,
 \end{align*}
 where the last line uses $2^{\offset} = 2^{\log(d) + \log(1/\badcut)} = O(d/\eps)$.

Since $b_1(p, \eps) = \beta(p, 3\calL_p)$ and $b_3(p, \eps) = \beta(\calL(p), 3\calS_{\calL(p)})$, we directly get $\E[b_1(p, \eps) + b_3(p, \eps)] \leq O(d^2 \log(1/\eps) \eps) (\calL_p^2 + \calS_p^2)$. 
Bounding $b_2(p, \eps)$ is only slightly more complicated: $b_2(p, \eps) = \beta(p, \calL_p + \calS_p) + X(p, \eps)$, where $X(p, \eps) := 36  d  \calL_p^2 +  16  d \calS_p^2$ when $\level(p, 3(\calL_p + \calS_p)) > \log(3(\calL_p + \calS_p)) + \offset$. 
Therefore, it is enough to compute $\E[X(p, \eps)]$.

Properties of the decomposition in \cref{lem:talwar-decomp} ensure that $\level(p, 3(\calL_p + \calS_p)) > \log(3(\calL_p + \calS_p)) + \offset$ with probability at most $\badcut$, and therefore $\E[X(p, \eps)] \leq \badcut \cdot (36  d  \calL_p^2 +  16  d \calS_p^2)) = O( \eps d (\calL_p^2 + \calS_p^2))$. 

Therefore, $\E[b_2(p, \eps)] = \E[\beta(p, 3(\calL_p + \calS_p)] + \E[X(p, \eps)] \leq O(d^2 \log(1/\eps) \cdot \eps)(\calL_p^2 + \calS_p^2)$.

Summing these inequalities for $b_1, b_2, b_3$, and for all $p \in P$ we get $\E[\sum_{p \in P} b(p, \eps)] \leq O( d^2  \log(1/\eps) \cdot \eps)(\cost(P, \calL)+ \cost(P, \calS))$. Markov's inequality concludes the lemma. 
\end{proof}

\subsubsection{Construction of $ \calS^*$}

To construct $\calS^*$, we start from the solution $\opt'$ obtained from \cref{lem:coststep1} and apply the two following steps.

For any $\ell \in \calL$, define  $f_\ell$ to be its closest center in $\psi(\ell)$, breaking ties arbitrarily. 
 We 
start with $ \calS^* = \opt'$ obtained from \cref{lem:coststep1}. Note that for every $\ell\in 
 \calL^1 \cup \calL^{\geq 2}$, the second property of \cref{lem:coststep1} ensures that $f_\ell$ is a center of $\opt'$.  $ \calS^*$ is then transformed as follows:
  \begin{itemize}
  \item \textbf{Step 1.} 
    For each center $\ell \in B_\calD\setminus \calL^0$ (i.e., $\ell$ is badly cut w.r.t.\ $\calD$ and $\opt'$, and 
$\psi(\ell)\neq\emptyset$), replace $f_\ell$ by $\ell$ in~$ \calS^*$.
  \item \textbf{Step 2.}
    Add all centers of $ \calL^0$ badly cut w.r.t.\ $\calD$ and $\opt'$ to $ \calS^*$.
  \end{itemize}

The two key properties of that procedure are the following:
\begin{fact}[Claim 4.3 and 4.4 in \cite{jacm}]\label{fact:propSstar}
With probability $4/5$, $|\calS^*| \leq k$, so $\calS^*$ is a valid solution, and $\cost(P,  \calS^*) \leq (1+\eps)\cost(P, \opt) + \eps \cost(P, \calL)$.
 Furthermore, for any point $p$,
  it holds that $\dist(p, \calS^*) \leq 2\dist(p, \opt') + \dist(p, \calL)$.  
\end{fact}
In addition, the construction of $\calS^*$ ensures that:
\begin{fact}\label{fact:ineq}
For any point $p$,
\begin{align*}
\dist(p, \tilde p) + \dist(\tilde p, \calS^*) &\leq 3\calL_p + 2\opt'_p\\
\dist(p, \calS^*(\calL(p))) &\leq 3\calL_p + 2\opt'_p.
\end{align*}

Furthermore, for any center $\ell \in \calL$, it holds that $\dist(\ell, \calS^*) \leq 2\dist(\ell, \opt')$.
\end{fact}
\begin{proof}
The first inequality follows from the fact that $\dist(\tilde p, \calS^*) \leq \dist(\tilde p, \calS^*(p)) \leq \dist(p, \tilde p) + \dist(p, \calS^*(p))$ and \cref{fact:propSstar}. 
The second follows similar lines: 
$\dist(p, \calS^*(\calL(p))) \leq \calL_p + \dist(\calL(p), \calS^*(\calL(p)) \leq \calL_p + \dist(\calL(p), \calS^*(p)) \leq 2\calL_p + \dist(p, \calS^*) \leq 3\calL_p + 2\opt'_p$.

For the third inequality, let $\ell$ be a center of $\calL$, and $\opt'(\ell)$ be its closest center in $\opt'$. Then, either $\opt'(\ell) \in \calS^*$, in which case we are done since $\dist(\ell, \calS^*) = \dist(\ell, \opt')$, or $\opt'(\ell)$ was removed in Step 1. 
In that case, it was replaced by the center $\ell'$ that is closest to $\opt'(\ell)$ in $\calL$: in particular, $\dist(\ell', \opt'(\ell)) \leq \dist(\ell, \opt'(\ell))$. Therefore, it holds that:
\begin{align*}
    \dist(\ell, \calS^*) &\leq \dist(\ell, \ell') \leq \dist(\ell, \opt'(\ell)) + \dist(\ell', \opt'(\ell)) 
    \\ &\leq 2\dist(\ell, \opt')\qedhere
\end{align*}
\end{proof}

We can now use this solution $ \calS^*$ to show $\calD$ has $\eps$-small distortion. 

\subsubsection{Bounding the cost of $\calS^*$}

\begin{lemma}
  \label{prop:structkmeans}
  $ \calS^*$ contains $B_{\calD}$, and with probability at least $4/5$ it holds that 
  \[\sum_{p \in P} \lpar \dist(p, \tilde p) + \dist(\tilde p,  \calS^*)\rpar^2 \leq (1+\eps)\cost(P, \opt) + \eps \cost(P, \calL)\]
\end{lemma}
\begin{proof}
By construction, $ \calS^*$ does contain $B_\calD$.
To simplify the equations, we write $p \bc$ when $p$ is badly cut w.r.t.\ $\calL$ (in which case $\tilde p = \calL(p)$). It holds that 
\begin{align}
\label{eq:decCostSstar}
    \sum_{p \in P} \lpar \dist(p, \tilde p) + \dist(\tilde p,  \calS^*)\rpar^2 &\leq \sum_{p \in P} \dist( p,  \calS^*)^2 + \sum_{p \bc} \lpar \dist(p, \tilde p) + \dist(\tilde p,  \calS^*)\rpar^2.
\end{align}
By \cref{fact:propSstar}, the first term is at most $(1+\eps)\cost(P, \opt) + \eps \cost(P, \calL)$: thus, we focus only on the second term.
For this, we use \cref{fact:ineq}:
\begin{align*}
    \sum_{p \bc} \lpar \dist(p, \tilde p) + \dist(\tilde p,  \calS^*)\rpar^2 & \leq \sum_{p \bc} \lpar 3\calL_p + 2\opt'_p\rpar^2\\
    &\leq 18\sum_{p \bc}  \calL_p^2 + {\opt'}_p^2
\end{align*}

Each point $p$ is badly cut with probability $\badcut$: therefore, in expectation, it holds that 
\begin{align*}
    \E[18\sum_{p \bc}   \calL_p^2 + {\opt'}_p^2] &\leq 18 \badcut \sum_{p \in P}  \calL_p^2 + {\opt'}_p^2\\
    &\leq 18\badcut (\cost(P, \calL) + \cost(P,  \opt'))
\end{align*}
Using Markov's inequality, it holds with probability $4/5$ that 
\begin{align*}
    \sum_{p \bc} \lpar \dist(p, \tilde p) + \dist(\tilde p,  \opt')\rpar^2 &\leq 18\cdot 5 \cdot \badcut (\cost(P, \calL) + \cost(P,  \opt')).
\end{align*}

Plugging this bound into \cref{eq:decCostSstar} and using \cref{lem:coststep1} to bound the cost of $\opt'$, we conclude:
\begin{align*}
    \sum_{p \in P} \lpar \dist(p, \tilde p) + \dist(\tilde p,  \calS^*)\rpar^2 &\leq (1+\eps)\cost(P, \opt) + \eps \cost(P, \calL)  + 18\cdot 5 \cdot \badcut (\cost(P, \calL) + \cost(P,  \opt')) \\
    &= (1+O(\eps))\cost(P, \opt) + O(\eps)\cost(P, \calL).\qedhere
\end{align*}
\end{proof}

\subsubsection{Bounding the detour}

We now show that the detour incurred by connecting $\tilde P$ to $\calS^*$ through portals of the decomposition $\calD$ is within the budget defined in \cref{def:budget}. 

\begin{lemma}
  \label{lem:detourkmeans}
Let $P \subset \R^d$, $\calL$ be a solution for $k$-means on $P$ and $\calD$ be a randomized decomposition. 
Let $\calS^*$ be the solution from \cref{prop:structkmeans}.

 Then, for any $p \in P$, there is a center $s \in \calS^*$ cut from $\tilde p$ at some level $i$ such that $(\dist(p, \tilde p) + \dist(\tilde p, s) + \eps 2^i)^2 \leq (\dist(p, \tilde p) + \dist(\tilde p, \calS^*))^2 + b(p, \eps)$.
 \end{lemma}
\begin{proof}

We recall that any point $x \in \R^d$ is badly cut w.r.t.\ $\calD$ and  a set  $S$ when the ball $B(x, 3\dist(x, S))$ is cut at level higher than $\log(9 \dist(x, S)) + \offset$. 
We recall our construction of $\tilde P$ and $\calS^*$: When a point $p$ is badly cut w.r.t $\calD$ and $\calL$, $\tilde p = \calL(p)$; when a center $\ell \in \calL$ is badly cut w.r.t $\calD$ and $\opt'$, $\ell \in \calS^*$.

To show the lemma, we make a case distinction, first according to whether $p$ is badly cut or not w.r.t.~$\calD$ and $\calL$. 
In the case where it is badly cut, it holds that $\tilde p = \calL(p)$. Then, we have the following cases.

\begin{enumerate}
 \item If $\calL(p)$ is badly cut w.r.t.\ $\calD$ and $\opt'$, then $\calL(p) \in \calS^*$: therefore, 
$\tilde p = \calL(p) =  \calS^*(p)$, and $\tilde p$ and $ \calS^*(p)$ are not cut, e.g. $i=-\infty$. Therefore, $s = \calL(p)$ satisfies the lemma, with $(\dist(p, \tilde p) + \dist(\tilde p, s) + \eps 2^i)^2 = \dist(p, \tilde p)^2$. 
 \item $\calL(p)$ is not badly cut w.r.t.\ $\calD$ and $\opt'$. Let $s = \calS^*(\calL(p))$, and $i$ be the level at which $\calL(p)$ and $s$ are cut. We write 
 $(\dist(p, \tilde p) + \dist(\tilde p, s) + \eps 2^i)^2 = (\calL_p + \dist(\tilde p, \calS^*))^2 + 2\eps  2^i \cdot \calL_p + 2 \eps 2^i\cdot \calS^*_{\calL(p)} + \eps^2 2^{2i}$. 
 
 By \cref{fact:ineq}, $s \in B(\calL(p), 2\opt'_{\calL(p)})$, and so
  $i \leq i' := \level\lpar \calL(p),  3\opt'_{\calL(p)}\rpar$. 
 Combined with the fact that $\calL(p)$ is not badly cut,  it therefore holds that 
 $2 \eps 2^i\cdot \calS^*_{\calL(p)} + \eps^2 2^{2i} \leq 2 \eps 2^{i'} \cdot 2\opt'_{\calL(p)} + \eps^2 2^{2i'} \leq b_3(p)$. 
 Thus, it is enough to bound $2  \eps 2^i \cdot \calL_p$. For this, we refine our case distinction, according to whether the ball $B(p, 3(\calL_p + \opt'_p))$ is badly cut or not, and use the budget from $b_2(p)$. 
 \begin{enumerate}
  \item If the ball is not badly cut, then we remark that \cref{fact:ineq} ensures  that both $s$ and $\calL(p)$ are in $B(p, 3(\calL_p + \opt'_p))$, and therefore $i \leq i'' := \level(p, 3(\calL_p + \opt'_p))$. Therefore, 
 $2  \eps 2^i \cdot\calL_p \leq \eps 2^{i''} 3(\calL_p + \opt'_p) \leq b_2(p)$.
  \item If this ball is badly cut, then we use $i \leq i' \leq \log(3\opt'_{\calL(p)}) + \offset$  in addition to $\opt'_{\calL(p)} \leq \calL_p + \opt'_p$ as follows:  
 $2  \eps 2^i \cdot\calL_p \leq 6 \eps \cdot \opt'_{\calL(p)} \cdot 2^{\offset} \calL_p \leq 6d (\calL_p + \opt'_p) \calL_p \leq b_2(p)$. 
 \end{enumerate}
 Combining these cases, we get that when $\calL(p)$ is not badly cut, $(\dist(p, \tilde p) + \dist(\tilde p, s) + \eps 2^i)^2 \leq (\calL_p + \dist(\tilde p, \calS^*))^2 + b_2(p) + b_3(p)$.
  \end{enumerate}

 We now turn to the case where $p$ is not badly cut. In that case, $\tilde p = p$: we therefore only need to bound $\lpar \dist( p, \calS^*) + \eps 2^i\rpar^2$. 
 Furthermore, the ball $B(p, 3\calL_p)$ is cut at level $\level(p, 3\calL_p) \leq \log(3\calL_p)+\offset$. 
We make a case distinction according to whether the ball $B(p, 3(\calL_p + \opt'_p))$  is badly cut or not. Note that, by \cref{fact:propSstar}, $ \calS^*(p) \in B(p, 3(\calL_p + \opt'_p))$.
 \begin{enumerate}
  \item If $B(p, 3(\calL_p + \opt'_p))$ is not badly cut, then $p$ and $s = \calS^*(p)$ are cut at level at most $i = \level(p, 3(\calL_p + \opt'_p))$, and $s$ satisfies the lemma: $(\dist(p, s) + \eps 2^i)^2 \leq \dist(p, \calS^*)^2 + b_2(p)$.
  \item If $B(p, 3(\calL_p + \opt'_p))$ is badly cut, we show that $s= \calS^*(\calL(p))$ satisfies the lemma. In particular,
  \begin{enumerate}
      \item if $\calL(p)$ is badly cut, then $s = \calL(p)$; in that case, since $p$ is not badly cut, $p$ and $s$ are cut at level at most $i = \level(p, 3\calL_p)$ and $(\dist(p, s) + \eps 2^i)^2 \leq \dist(p, \calS^*)^2 + b_1(p) + b_2(p)$.
      \item if $\calL(p)$ is not badly cut, then one can assign $p$ to $s=\calS^*(\calL(p))$. The key observation here is that they are cut at level at most $i = \max\lpar \level(p, \calL_p), \level(\calL(p),  3\opt'_{\calL(p)}) \rpar$: indeed, from \cref{fact:ineq} $s$ is in the ball $B(\calL(p), 3\opt'_{\calL(p)})$, therefore $s$ and $\calL(p)$ are in the same cell at level $i$; similarly, $p$ and $\calL(p)$ are at the same cell at level $i$, and, therefore, $p$ and $s$ are cut at level at most $i$.  Additionally, from \cref{fact:ineq}, $p$ and $s$ satisfy $\dist(p, s) \leq 3\calL_p + 2\opt'_p$.    
 Thus, using that for any $a, b, c$, $(a+b+c)^2 \leq 3(a^2 + b^2 + c^2)$:
       \begin{align*}
           (\dist(\tilde p, s) + \eps 2^i)^2 &\leq 
           (3\calL_p + 2\opt'_p + \eps 2^i)^2\\
           &\leq 27 \calL_p^2 + 12 {\opt'}_p^2 + 3\eps^2 2^{2i}.
       \end{align*}
    Now, since $B(p, 3(\calL_p + \opt'_p))$ is badly cut, $27 \calL_p^2 + 12 {\opt'}_p^2 \leq b_2(p)$. Furthermore, by the upper bound on $i$ we have $3\eps^2 2^{2i} \leq b_1(p) + b_3(p)$, as both $p$ and $\calL(p)$ are not badly cut. 
    Therefore, $(\dist(\tilde p, s) + \eps 2^i)^2 \leq b_1(p)+ b_2(p) + b_3(p)$, which concludes the case and the lemma.\qedhere
  \end{enumerate}
\end{enumerate}
%
\end{proof}

\subsection{The Algorithm: how to use the Structure Theorem}\label{sec:algo}

Suppose that $\calD$ has $\rho$-small-distortion, for $\rho = O(\eps / \log(1/\eps))$ (where the $O$ hides dependencies in $d$), and let $\calS^*$ be the solution provided by \cref{def:smallDist}. 
As explained previously, our algorithm will be a standard dynamic program working as follows. The decomposition $\calD$ will be augmented with a set of portals, as defined in \cref{lem:talwar-decomp}, with $\rho = \eps / \log(1/\eps)$. 
The dynamic program (described later in greater detail) computes the best portal-respecting solution, namely, the solution where each path connecting a client to a center crosses boundaries of $\calD$ only at portals. More precisely, it computes the portal-respecting solution with (almost) smallest $\tilde{\cost}(P, \calS) := \sum_p (\dist(p, \tilde p) + \dist(\tilde p, \calS))^2$

\begin{theorem}\label{thm:DP}
    There exists an algorithm running in time $(2/\eps)^{O\lpar \log(1/\eps)/\eps)^{d-1}\rpar} \cdot n \log^3 n$ that computes a portal-respecting solution $\calS$ such that $\tilde {\cost}(P, \calS) \leq (1+\eps) \min_{\calS'} \tilde{\cost}(P, \calS')$.
\end{theorem}

The preciseness and nestedness of portal sets ensure that, when two points $x$ and $y$ are cut at level $i$, then the portal-respecting path between $x$ and $y$ has length $\|x-y\| + \rho 2^{i+1}$. 
Furthermore, the algorithm from \cref{thm:DP} computes a solution at least as good as $\calS^*$: by Property \ref{prop:det} of \cref{def:smallDist}, the portal-respecting $\tilde{\cost}$ of $\calS^*$ is at most $\sum_p (\dist(p, \tilde p) + \dist(\tilde p, \calS^*))^2 + b(p, \rho)$.  Property \ref{prop:budget} ensures that the total budget is at most $\rho (\cost(P, \opt) + \cost(P, \calL)) = O(\rho) \cost(P, \opt)$, since $\calL$ is an $O(1)$-approximation. 
Combined with Property \ref{prop:costSstar} of \cref{def:smallDist}, this shows that the portal-respecting $\tilde{\cost}$ of the solution computed by the algorithm is at most $(1+O(\eps))\cost(P, \opt)$. 

To wrap up, we note that the portal respecting $\tilde{\cost}$ is an upper bound on cost: going in straight lines instead of going through portals only improves the cost, and for any point $p$, $\dist(p, \calS) \leq \dist(p, \tilde p) + \dist(\tilde p, \calS)$. Therefore, the solution computed by the dynamic program is a $(1+O(\eps))$-approximation for $k$-means.

\paragraph*{Dynamic Program: proof of \cref{thm:DP}}

The algorithm used to show \cref{thm:DP} is a quite standard dynamic program (DP), described in detail in \cite{jacm}. We give here a brief overview.
For the sake of the DP, we will represent the decomposition $\calD$ as a binary tree, instead of a $2^d$-ary tree.\footnote{Technically, in $\calD$ each cell is decomposed into $2^d$ children, obtained by cutting it with $d$ hyperplanes $H_1, ..., H_d$. Instead, we represent the decomposition by first cutting along $H_1$, creating two children, themselves cut along $H_2$, creating two children for each of them, etc.}

A \textit{configuration} consists of the following parameters: 
\begin{itemize}
    \item a node of the quadtree; let $D$ be its  diameter,
    \item for each portal $p_i$ of the node, two values $\ell_i$ and $s_i$, multiples of $\eps D$ in the range $\lbra 0, D/\eps + 1/\eps\rbra$,
    \item a boolean flag, indicating whether the closest center to the cell is at distance at least $D/\eps$ or not,
    \item a target cost $c_0$, power of $(1+\eps/\log n)$ in the range $\lbra \cost(\calL)/n, (1+\eps) \cost(\calL) \rbra$.
\end{itemize}
The value $\ell_i$ encodes to the distance from portal $p_i$ to the closest center inside the part, and $s_i$ the distance to the closest center outside of the part.

The DP computes, for each configuration, the smallest number of centers necessary to place in the cell to match the portal parameters and have a solution with $\tilde{\cost}$ at most $c_0$.

The first base case of the DP is when the node is a leaf: then, as explained in \cite{jacm}, there is a single point from $\tilde P$ in the cell, and the two options are either to open a center in the cell to have $\tilde{\cost} = 0$, or to connect to a center outside via a portal-respecting path, whose distance is provided by the values of $s_i$.

The other base case of the DP is when the boolean flag indicates that the closest center is at distance more than $D/\eps$, then every point in the cell can be assumed to be at the center of the cell, as if there was a single point. Then, the value of the configuration is computed as for a leaf.

In the non-base case, the recursion is also standard: for a given configuration, 
the DP enumerates all possible compatible configurations for the two children of the node, namely: all configurations where $\ell_i$ and $s_i$ are the same when the portals are the same, respects the triangle inequality otherwise, and the sum of the two estimated costs is equal to $c_0$ (up to the rounding).

The proof of correctness can be found in Section 4.1 of \cite{jacm}. The complexity is improved, as we have fewer portals: by \cref{lem:talwar-decomp} with $\rho = \eps/\log(1/\eps)$, there are $(\log(1/\eps) / \eps)^{d-1}$ of them. For each portal, there are $1/\eps$ choices for the values of $\ell_i$ and $s_i$. Furthermore, there are $O(\log^2 n)$ possible values for $c_0$. 
Last, there are $O(nd \log n)$ nodes in the tree, as each leaf contains at least one point from $\tilde P$, and the depth of the tree is $O(d \log n)$, since $|\calD| \leq \log(\diam(P))$, $\diam(P) = \poly(n)$, and the binary tree has $d$ times more levels than the decomposition $\calD$. 
Therefore, the complexity of the DP is $(2/\eps)^{O\lpar \log(1/\eps)/\eps)^{d-1}\rpar} dn \log^3 n$. 

\section{Fine-Grained Hardness of Approximation in Low Dimensional Spaces}\label{sec:infLD}

The goal of this section is to prove \cref{thm:infLD}, which we restate here for convenience:
\infLD*

Our hardness result is under the Gap-ETH hypothesis:

\begin{definition}[Gap Exponential Time Hypothesis (Gap-ETH)~\cite{D16,MR16}] \label{hyp:gap-eth}
There exist constants $\theta,\delta> 0$ such that any algorithm that, given as input a 3-CNF formula $\varphi$ on $n$ variables and $O(n)$ clauses, can distinguish between $\sat(\varphi)=1$ and $\sat(\varphi)<1-\delta$, must run in time at least $2^{\theta n}$.
\end{definition}

In order to preserve the gap in our reduction, we need to ensure each variable appears in only few clauses. For this, 3-CNF formulas are not enough, but the following lemma shows that one can assume each variable appears in at most $3$ clauses, and each clause has size at most 3 (this is called a $(3, 3)$-CNF):
\begin{lemma}[Corollary 6.1 in \cite{Kisfaludi-BakNW21}]\label{lem:gap-eth}
    There exist constants $\delta, \gamma > 0$ such that there is no algorithm running in time $2^{\gamma n}$ that, given a
(3,3)-CNF formula $\phi$ on $n$ variables and $m$ clauses, can distinguish between the cases where (i) $\phi$
is satisfiable or (ii) every assignment violates at least $\delta m$ clauses, unless Gap-ETH fails.
\end{lemma}

Next, the result we use from \cite{BergBKMZ20} is the following. For any (3, 3)-CNF formula $\phi$ on $n$ variables and $m$ clauses, there exists a graph $G = (V,E)$ and an embedding $\Pi : V \rightarrow \mathbb{R}^d$ with the following properties:
\begin{itemize}
    \item $|E| = \Theta(n^{d/(d-1)})$, 
    \item $E$ can be partitioned into $E_1$ and $E_2$, with, for each edge $(u,v) \in E_1$, $\|\Pi(u) - \Pi(v)\| = 2$, and for any edge $(u, v) \in E_2$, $\|\Pi(u) - \Pi(v)\| = 2\sqrt{2}$.
    \item For each edge $(u,v) \in E$, the squared Euclidean distance from the midpoint $(\Pi(u) + \Pi(v))/2$ to any other non-incident embedded vertex in $\Pi(V)$ is at least $4$.
\end{itemize}

\begin{proof}[Proof Outline of Embedding Properties from \cite{BergBKMZ20}] 
Based on the framework by de Berg et al., we construct the graph $G$ and its embedding as follows. First, the incidence graph of the formula $\phi$ is drawn in a grid, and using the cube wiring theorem (Theorem 3.3 in \cite{BergBKMZ20}), it is routed into a $d$-dimensional grid of side length $O(n^{\frac{1}{d-1}})$. The total number of vertices in this induced grid graph is $O(n^{\frac{d}{d-1}})$, meaning the number of edges $|E|$ is bounded by $\Theta(n^{\frac{d}{d-1}})$. 

The construction replaces vertices and edges with specific components:
\begin{itemize}
    \item Each variable $v_i$ is replaced by a variable gadget, which is a cycle of length eight.
    \item Each clause containing three literals is replaced by a cycle of length nine.
    \item Each clause containing two literals is replaced by a single edge.
    \item The wires connecting these gadgets are represented by simple paths of odd length.
\end{itemize}

This graph is embedded into the augmented $d$-dimensional grid, which consists of the standard grid along with diagonal edges added to two-dimensional faces. To establish the exact spatial distances, the embedding applies an affine transformation that skews the space such that standard grid edges ($E_2$) have a length of exactly $2\sqrt{2}$ (squared distance $8$) and the utilized diagonal edges ($E_1$) have a length of exactly $2$ (squared distance $4$). Furthermore, this transformation guarantees that the squared distance between the midpoint of any edge and any non-incident vertex is bounded below by exactly $4$.
\end{proof}

Furthermore, the graph $G$ relates to the formula $\phi$ through its local vertex covers. Let $W$ be the global set of wires in $G$, and let a \emph{part} refer to any variable gadget, clause gadget, or wire.
\begin{itemize}
    \item For any wire $w \in W$ of length $2t_w + 1$, a valid minimum local cover must contain exactly $t_w$ of its inner vertices. 
    \item For any clause $C$, let $G_C$ be the subgraph induced by the gadget of $C$, the gadgets of its variables, and the connecting wires $W_C \subseteq W$. 
    \item We can decode a satisfying assignment for a 3-literal clause from any valid minimum local vertex cover of $G_C$ that uses exactly $5$ vertices on the clause cycle, $4$ on each variable cycle, and $t_w$ on each connecting wire. 
    \item Similarly, for a 2-literal clause, we can decode a satisfying assignment from a minimum local cover using exactly $1$ vertex on the clause segment, $4$ on each variable cycle, and $t_w$ on each wire. 
    \item Globally, if $\phi$ is satisfiable (with $n_v$ variables, $n_2$ clauses of size two, and $n_3$ clauses of size three), the minimum vertex cover of $G$ has size exactly $k :=  4n_v + 5n_3 + n_2 + \sum_{w \in W} t_w$. 
\end{itemize}
Note that in graph $G$, any minimum vertex cover corresponding to a satisfying assignment tightly packs exactly 4 vertices per variable cycle, 5 vertices per 3-literal clause cycle, and 1 vertex per 2-literal clause segment. Thus, from a global minimum vertex cover of size $k$ on $G$, one can efficiently recover a satisfying assignment.

From such a graph $G$, we construct a discrete clustering instance with set of clients $P$ and set of candidate centers $\calC$ as follows: $\calC = \{\Pi(v): v \in V\}$ is the embedding of all vertices, and 
$P$ is made of the union of all mid-points of edges of $E$, namely, for any edge $(u,v) \in E$, the midpoint $p_{uv} := (\Pi(u) + \Pi(v))/2$ is in $P$. 

The properties of the embedding ensure the following facts. 
For any edge $e \in E_1$, the closest candidate center to $p_e$ is at a distance of $1$ (squared distance $1$), and for any edge $e \in E_2$, the closest candidate center to $p_e$ is at a distance of $\sqrt{2}$ (squared distance $2$). Therefore, if $G$ has a vertex cover of size $k$, the optimal discrete $k$-means cost is exactly $|E_1| + 2|E_2|$, and the optimal discrete $k$-median cost is exactly $|E_1| + \sqrt{2}|E_2|$.

Our main result is the following, which relates the cost of a discrete clustering solution to the formula $\phi$:
\begin{lemma}\label{lem:decodeD}
Let $\calS$ be a solution to discrete $k$-means (resp. discrete $k$-median) clustering on $G=(V, E)$ for $k = 4n_v + 5n_3 + n_2 + \sum_{w \in W} t_w$ with cost at most $(1+\eps)\cdot (|E_1| + 2|E_2|)$ (resp. $(1+\eps)\cdot (|E_1| + \sqrt{2}|E_2|)$). Then, we can construct in time polynomial in $|V|$ an assignment of variables satisfying at least $m - c'\eps|E|$ many clauses of $\phi$, for some constant $c' > 0$.
\end{lemma}
\begin{proof}
We see the centers of $\calS$ as covering edges of the graph $G$: centers of $\calS$ correspond to vertices in that graph, and clients to edge midpoints.
We say a part is \emph{destroyed} by $\calS$ when $\calS$ does not induce a valid local minimum vertex cover on that part. 

We can relate the number of destroyed parts with the number of non-covered edges, in order to link the cost of the clustering solution with a valid assignment. We classify the destroyed parts into three types:
\begin{itemize}
    \item \textbf{Type (A):} Parts destroyed because $\calS$ places too few vertices on them (deficiency).
    \item \textbf{Type (B):} Parts destroyed because $\calS$ places too many vertices on them (surplus).
    \item \textbf{Type (C):} Parts that receive the exact correct number of vertices, but leave at least 1 internal edge uncovered.
\end{itemize}

Let $D$ be the total vertex deficiency across all Type (A) parts, and $S$ be the total surplus of vertices in Type (B) parts. Since exactly $k$ centers are used globally, the total deficiency must equal the total surplus, meaning $D = S$. Let $k_A, k_B,$ and $k_C$ be the number of parts of each type, so the total number of destroyed parts is $\delta_{parts} = k_A + k_B + k_C$. Note that $D \ge k_A$ and $S \ge k_B$.

Every part $P$ in our construction is a subgraph where each vertex has an internal degree of at most 2. Let $E_P$ be the number of internal edges of a part and $k_P$ be its optimal local vertex cover size. If a Type (A) part is deficient by $d \ge 1$ vertices, it uses $k_P - d$ vertices, which can cover at most $2(k_P - d)$ edges. Thus, the number of completely uncovered internal edges is at least $E_P - 2(k_P - d) = 2d + (E_P - 2k_P)$. 

We can evaluate this lower bound by checking $E_P$ and $k_P$ for each gadget type. For a variable gadget, which is a cycle of length eight, we have $E_P = 8$ and an optimal cover of $k_P = 4$, yielding $2d + (8 - 8) = 2d \ge d$ uncovered edges. For a clause with three literals, the gadget is a cycle of length nine with an optimal cover of $k_P = 5$, resulting in $2d + (9 - 10) = 2d - 1$ uncovered edges. Because a Type (A) part is deficient by definition, we know $d \ge 1$, which guarantees $2d - 1 \ge d$. This exact same arithmetic applies to the remaining parts: the wire gadgets are simple paths of odd length (where $E_P = 2t_w - 1$ and $k_P = t_w$), and the clauses with exactly two literals use an edge as a clause gadget (where $E_P = 1$ and $k_P = 1$). In both of these cases, the formula again evaluates to $2d - 1 \ge d$. 

Consequently, a Type (A) part missing $d$ vertices always leaves at least $d$ internal edges completely uncovered. Summing this globally, Type (A) parts contribute at least $D \ge k_A$ uncovered internal edges. Furthermore, by definition, every Type (C) part does not induce a valid local vertex cover, so it must leave at least $1$ internal edge uncovered, contributing at least $k_C$ uncovered edges overall.

Since internal edges of different parts are mutually disjoint, we can sum them without double counting. The total number of uncovered edges $U$ is bounded by:
\[ U \ge D + k_C \]
Since $D = S$, we know $D = \frac{D + S}{2} \ge \frac{k_A + k_B}{2}$. Substituting this yields:
\[ U \ge \frac{k_A + k_B}{2} + k_C = \frac{k_A + k_B + 2k_C}{2} \ge \frac{\delta_{parts}}{2} \]
Thus, there are at least $\delta_{parts}/2$ uncovered edges in $G$.

We say a clause is \emph{bad} if it is destroyed, or one of the variables in the clause is destroyed, or one of its variables/wires is destroyed, or if there is a mismatch between them (which leaves a connecting edge uncovered).

Let $\beta$ be the total number of bad clauses. We can partition these into $\beta_{\text{dest}}$ clauses that contain at least one destroyed part, and $\beta_{\text{mis}}$ clauses that have no destroyed parts but contain a mismatch (so $\beta = \beta_{\text{dest}} + \beta_{\text{mis}}$). Since a single destroyed part contributes to at most 3 bad clauses, the $\beta_{\text{dest}}$ clauses correspond to at least $\beta_{\text{dest}}/3$ destroyed parts. As seen earlier, this yields at least $\frac{1}{2}(\beta_{\text{dest}}/3) = \beta_{\text{dest}}/6$ uncovered internal edges. The $\beta_{\text{mis}}$ mismatched clauses have completely valid local parts but there must be at least one uncovered connecting edge per clause (yielding $\ge \beta_{\text{mis}}$ uncovered connecting edges). Because internal edges and connecting edges are mutually disjoint, the total number of uncovered edges $U$ is at least $\frac{\beta_{\text{dest}}}{6} + \beta_{\text{mis}} \ge \frac{\beta_{\text{dest}} + \beta_{\text{mis}}}{6} = \frac{\beta}{6}$.

By the embedding properties, an uncovered edge has its client (the edge midpoint) at a squared distance of at least $4$ from any selected center. For the discrete $k$-means problem, if it were covered locally, the cost would be $1$ (for $E_1$) or $2$ (for $E_2$). Since the maximum standard cost of a locally covered edge is $2$, an uncovered edge introduces a strict additive penalty of at least $4 - 2 = 2$. Consequently, the total cost of the discrete $k$-means solution is bounded below by $|E_1| + 2|E_2| + 2\left(\frac{\beta}{6}\right) = |E_1| + 2|E_2| + \frac{\beta}{3}$.

For the discrete $k$-median problem (which minimizes the sum of standard, un-squared distances), the analysis is entirely analogous. An uncovered edge midpoint is at a distance of at least $\sqrt{4} = 2$ from any selected center, whereas a locally covered edge has a cost of $1$ (for $E_1$) or $\sqrt{2}$ (for $E_2$). Since the maximum standard cost of a locally covered edge is $\sqrt{2}$, an uncovered edge introduces a strict additive penalty of at least $2 - \sqrt{2} > 0.58$. Consequently, the total cost of the discrete $k$-median solution is bounded below by $|E_1| + \sqrt{2}|E_2| + 0.58\left(\frac{\beta}{6}\right) \ge |E_1| + \sqrt{2}|E_2| + \frac{\beta}{11}$.

Conversely, we are given a solution $\calS$ with cost at most $(1+\eps) \cdot \text{OPT}$. 
For $k$-means, this implies the extra penalty is bounded by $\frac{\beta}{3} \leq \eps(2|E_1| + 2|E_2|) = 2\eps |E|$, simplifying to $\beta \leq 6\eps |E|$. 
For $k$-median, the penalty is bounded by $\frac{\beta}{11} \leq \eps(\sqrt{2}|E_1| + \sqrt{2}|E_2|) = \sqrt{2}\eps |E|$, simplifying to $\beta \leq 11\sqrt{2}\eps |E|$. 
In both cases, we can bound the number of bad clauses by $c'\eps|E|$ for a constant $c' > 0$ (e.g., $c'=6$ for $k$-means and $c'=16$ for $k$-median).

Finally, we can efficiently compute a satisfying assignment for all non-bad clauses. For this, we rely on the result of \cite{BergBKMZ20}, who showed how to decode a valid assignment for $\phi$ given a valid vertex cover of $G$. Restricted to the clauses that are not bad, the subgraphs act as a valid local vertex cover, providing a valid assignment for all $m - c'\eps |E|$ non-bad clauses.
\end{proof}

Using \cref{lem:decodeD}, we can conclude the proof of our \cref{thm:infLD}.
\begin{proof}
Let $\delta, \gamma$ be the constants from Lemma~\ref{lem:gap-eth}. Note that for a (3,3)-CNF formula, the number of clauses $m$ is linearly related to the number of variables $n$ ($m = \Theta(n)$).

Assume for the sake of contradiction the existence of an approximation scheme that takes $\eps > 0$ as input and computes a $(1+\eps)$-approximation to discrete $k$-means (or $k$-median) in time $2^{c(1/\eps)^{d-1}} \poly(N)$, where $N = |V| = \Theta(n^{d/(d-1)})$ is the size of the instance, and $c > 0$ is a sufficiently small constant to be determined to reach a contradiction.

Fix a (3,3)-CNF formula $\phi$ with $n$ variables and $m$ clauses, and build the graph $G$ and the corresponding clustering instance as described. We set our approximation parameter $\eps$ such that the number of bad clauses $c'\eps|E|$ is exactly $\frac{\delta}{2} m$. Since $|E| = \Theta(n^{d/(d-1)})$ and $m = \Theta(n)$, this dictates that we set:
\[ \eps = \Theta\left(\frac{n}{n^{d/(d-1)}}\right) = \Theta(n^{-1/(d-1)}) \]
Note that this implies $(1/\eps)^{d-1} = \Theta(n)$. 

By \cref{lem:decodeD}, if we run our hypothetical approximation scheme with this specific $\eps$, we can decode in polynomial time an assignment satisfying all clauses that are not bad. 
Therefore, the embedding ensures the following: if $\phi$ is satisfiable, the optimal discrete $k$-means (or $k$-median) cost is exactly $\text{OPT}$. If the scheme returns a solution of cost at most $(1+\eps)\cdot\text{OPT}$, then we recover an assignment satisfying all but $\frac{\delta}{2} m < \delta m$ clauses. 

The time taken by the scheme to run with this $\eps$ is:
\[ 2^{c\left(1/\eps\right)^{d-1}} \poly(N) =  2^{c \cdot \Theta(n)} \poly(n). \]
By choosing the constant $c$ in the algorithm's running time to be sufficiently small, this overall time becomes strictly less than $2^{\gamma n}$, contradicting Gap-ETH. 
\end{proof}

\bibliographystyle{plainurl}
\bibliography{biblio}{}

\appendix

\section{Algorithm for $k$-Median in Low Dimensions}\label{app:k-med}

For $k$-median, since distances are not squared, we need to slightly adapt our budget and argument. Given this modified budget, the structure of the algorithm is almost identical; thus, we only present the proof of the structure theorem for $k$-median here.

\begin{definition}\label{def:budgetKmed}
Let $P\subset \R^d$, $\calD$ be 
a hierarchical decomposition of $P$, $\eps > 0$, $\calL$ and $\calS$ be two 
solutions to $k$-median on $P$. 

For a ball $B(x, r)$, its \emph{detour} with respect to $B$ and $\calD$ is $\deto_\calD(x, r) = \eps 2^{\level(x, r)}$.

A point $p\in P$ has a $k$-median budget with respect to $\calL$ and $\calS$ of $b(p, \eps) = b_1(p, \eps) + b_2(p, \eps) + b_3(p, \eps)$, with
\begin{align*}
    b_1(p, \eps) &= \begin{cases} 
 \deto_\calD(p, 3\calL_p)\text{ if } \level(p, 3\calL_p)\leq \log(3\calL_p) + \offset \\
  0\text{ otherwise }
\end{cases}
\\
b_2(p, \eps) &= \begin{cases} 
\deto_\calD\lpar p, 3(\calL_p + \calS_p)\rpar
 \text{ if } \level(p, 3(\calL_p + \calS_p))\leq \log(3(\calL_p + \calS_p)) + \offset \\
 3\calL_p + 2\calS_p \text{ otherwise }
\end{cases}
\\
b_3(p, \eps) &= \begin{cases} 
\deto_\calD(\calL(p),  3\calS_{\calL(p)})
 \text{ if } \level(\calL(p), 3\calS_{\calL(p)})\leq \log(3\calS_{\calL(p)}) + \offset \\
 0 \text{ otherwise }
\end{cases}
\end{align*}
\end{definition}

Our goal is to show, similar to \cref{thm:smallDist}, that $\calD$ has small distortion for $k$-median, defined as follows:
\begin{definition}\label{def:smallDistKmed}
For $\eps  > 0$, we say that $\calD$ has $\eps$-\emph{small distortion} for $k$-median if 
\begin{enumerate}
\item\label{prop:budgetKmed} the total $k$-median budget w.r.t.\ $\calL$ and $\opt'$ is bounded: $\sum_{p \in P} b(p, \eps) \leq \eps (\cost(P, \opt) + \cost(P, \calL))$.
\item\label{prop:costSstarKmed} there exists a solution $ \calS^*$ that contains $B_{\calD}$ 
with $\sum_{p \in P}  \dist(p, \tilde p) + \dist(\tilde p,  \calS^*) \leq (1+\eps)\cost(P, \opt) + \eps \cost(P, \calL)$,
\item\label{prop:detKmed} for each point $p \in P$, its detour to $\calS^*$ is at most $b(p, \eps)$, namely there is a center $s \in \calS^*$ cut from $\tilde p$ at level $i$ such that $\dist(p, \tilde p) + \dist(\tilde p, s) + \eps 2^i \leq \dist(p, \tilde p) + \dist(\tilde p, \calS^*) + b(p, \eps)$.
\end{enumerate}
\end{definition}

The solution $\calS^*$ is constructed exactly as in \cref{sec:constructS*}, with the only difference that $\calL$ and $\opt'$ are now solutions to $k$-median instead of $k$-means.

\paragraph{Bounding the Budget.}

As for $k$-means, we first show that the budget is bounded:

\begin{lemma}\label{lem:budgetBoundedKmed}
Fix two solutions $\calL$ and $\calS$, and let $b(p, \eps)$ be the budget of $p$ w.r.t.\ solutions $\calL$ and $\calS$. With probability  $2/3$ (over the randomness of $\calD$),   $\sum_p b(p, \eps) = O(d \log(1/\eps) \cdot \eps) (\cost(P, \calL) + \cost(P, \calS))$.
\end{lemma}
\begin{proof}
 We bound the expectation of the budget and then use Markov's inequality. Recall that any ball $B(x, r)$ is cut at level $i$ with probability $\sqrt{d} r / 2^i$, and is  badly cut with probability $\badcut$. Furthermore, when $B(x, r)$  is cut at level $i$, then $\deto_\calD(x, r) = \eps 2^{i}$. 
 
 To simplify the analysis slightly, we first analyze the expectation of the variable 
\[\beta(x, r) := \begin{cases}\deto_\calD(x, r) \text{ when } \level(x, r) \leq \log(r) + \offset\\
 0 \text{ otherwise.}\end{cases}\]
 To do so, we distinguish between the cases where $\level(x, r)$ is at most $\log(r)$, between $\log(r)$ and $\log(r) + \offset$, and at least $\log(r) + \offset$. This gives the following upper bound:
 \begin{align*}
     \E[\beta(x, r))]
     &\leq \sum_{i \leq \log(r)} \deto_\calD(x, r) + \sum_{i = \log(r)}^{\log(r) + \offset} \deto_\calD(x, r) \cdot \Pr[\calD \text{ cuts } B(p,r) \text{ at a level } i]\\
     &\leq \sum_{i \leq \log(r)} \eps 2^{i} + \sum_{i = \log(r)}^{\log(r) + \offset}\eps 2^{i} \cdot \frac{\sqrt{d} r}{2^i}\\
     &\leq 2 \eps r + \sum_{i = \log(r)}^{\log(r) + \offset} \eps \sqrt{d} r \\
     &\leq 2\eps r + \offset \cdot \eps \sqrt{d} r\\
     &= O(\sqrt{d}\log(d/\eps)) \eps r,
 \end{align*}
 where the last line uses $\offset =\log(d)/2 + \log(1/\badcut)$.

Since $b_1(p, \eps) = \beta(p, 3\calL_p)$ and $b_3(p, \eps) = \beta(\calL(p), 3\calS_{\calL(p)})$, we directly get $\E[b_1(p, \eps) + b_3(p, \eps)] \leq  O(\sqrt{d}\log(d/\eps)) \eps (\calL_p + \calS_p)$. 

Bounding $b_2(p, \eps)$ is only slightly more complicated: $b_2(p, \eps) = \beta(p, \calL_p + \calS_p) + X(p, \eps)$, where $X(p, \eps) := 3\calL_p + 2\calS_p$ when $\level(p, 3(\calL_p + \calS_p)) > \log(3(\calL_p + \calS_p)) + \offset$. 
Therefore, it is enough to compute $\E[X(p, \eps)]$.
Properties of the decomposition in \cref{lem:talwar-decomp} ensure that $\level(p, 3(\calL_p + \calS_p)) > \log(3(\calL_p + \calS_p)) + \offset$ with probability at most $\badcut$, and therefore $\E[X(p, \eps)] \leq \badcut \cdot (3\calL_p + 2\calS_p) = O(\eps)(\calL_p + \calS_p)$. 

Therefore, $\E[b_2(p, \eps)] = \E[\beta(p, 3(\calL_p + \calS_p)] + \E[X(p, \eps)] \leq O(\eps)(\calL_p + \calS_p)$.

Summing these inequalities for $b_1, b_2, b_3$ and all $p\in P$, we get $\E[\sum_{p \in P} b(p, \eps)] \leq O(\sqrt{d}\log(d/\eps)\eps)(\cost(P, \calL) + \cost(P, \calS))$. Markov's inequality concludes the lemma.
\end{proof}

\paragraph{Bounding the cost of $\calS^*$.}

Now, exactly as for $k$-means, the solution $\calS^*$ is near-optimal: we can use it to show the following lemma, similar to \cref{prop:structkmeans}. We repeat the proof for completeness.

\begin{lemma}
  \label{prop:structkmedian}
  $ \calS^*$ contains $B_{\calD}$, and with probability at least $4/5$ it holds that 
  \[\sum_{p \in P}  \dist(p, \tilde p) + \dist(\tilde p,  \calS^*) \leq (1+\eps)\cost(P, \opt) + \eps \cost(P, \calL)\]
\end{lemma}
\begin{proof}
By construction, $ \calS^*$ does contain $B_\calD$.
To simplify the equations, we write $p \bc$ when $p$ is badly cut w.r.t.\ $\calL$ (in which case $\tilde p = \calL(p)$). It holds that 
\begin{align}
\label{eq:decCostSstarKmed}
    \sum_{p \in P} \dist(p, \tilde p) + \dist(\tilde p,  \calS^*) &\leq \sum_{p \in P} \dist( p,  \calS^*) + \sum_{p \bc} \dist(p, \tilde p) + \dist(\tilde p,  \calS^*).
\end{align}
By properties of $ \calS^*$, the first term is at most $(1+\eps)\cost(P, \opt) + \eps \cost(P, \calL)$: thus, we focus only on the second term.
For this, we use \cref{fact:ineq}:
\begin{align*}
    \sum_{p \bc} \dist(p, \tilde p) + \dist(\tilde p,  \calS^*) & \leq \sum_{p \bc} 3\calL_p + 2\opt'_p
\end{align*}

Each point $p$ is badly cut with probability $\badcut$: therefore, in expectation, it holds that 
\begin{align*}
    \E[3\sum_{p \bc}  \calL_p + \opt'_p] &\leq 3 \badcut \sum_{p \in P}\calL_p + \opt'_p \leq 3\badcut (\cost(P, \calL) + \cost(P,  \opt'))
\end{align*}
Using Markov's inequality, it holds with probability $4/5$ that 
\begin{align*}
    \sum_{p \bc}  \dist(p, \tilde p) + \dist(\tilde p,  \opt') &\leq 3\cdot 5 \cdot \badcut (\cost(P, \calL) + \cost(P,  \opt')).
\end{align*}

Plugging this bound into \cref{eq:decCostSstarKmed} and using \cref{lem:coststep1} to bound the cost of $\opt'$, we conclude:
\begin{align*}
    \sum_{p \in P} \dist(p, \tilde p) + \dist(\tilde p,  \calS^*) &\leq (1+\eps)\cost(P, \opt) + \eps \cost(P, \calL)  + 3\cdot 5 \cdot \badcut (\cost(P, \calL) + \cost(P,  \opt')) \\
    &= (1+O(\eps))\cost(P, \opt) + O(\eps)\cost(P, \calL).
\end{align*}
\end{proof}

\paragraph{Bounding the detour.}
As for $k$-means, we now conclude the structure theorem showing that the detour incurred by connecting $\tilde P$ to $\calS^*$ through portals of the decomposition $\calD$ is within the budget defined in \cref{def:budgetKmed}. Again, we repeat the proof only for completeness: the case disjunction is slightly simpler, as it does not involve squares.

\begin{lemma}
  \label{lem:detourkmedians}
Let $P \subset \R^d$, $\calL$ be a solution for $k$-median on $P$ and $\calD$ be a randomized decomposition. 
Let $\calS^*$ be the solution from \cref{prop:structkmeans}.

 Then, for any $p \in P$, there is a center $s \in \calS^*$ cut from $\tilde p$ at level $i$ such that $\dist(p, \tilde p) + \dist(\tilde p, s) + \eps 2^i \leq \dist(p, \tilde p) + \dist(\tilde p, \calS^*) + b(p, \eps)$.
 \end{lemma}
\begin{proof}
We recall that any point $x \in \R^d$ is badly cut w.r.t.\ $\calD$ and a set $S$ when the ball $B(x, 3\dist(x, S))$ is cut at level higher than $\log(9 \dist(x, S)) + \offset$. 
We recall our construction of $\tilde P$ and $\calS^*$: When a point $p$ is badly cut w.r.t $\calD$ and $\calL$, $\tilde p = \calL(p)$; when a center $\ell \in \calL$ is badly cut  w.r.t $\calD$ and $\opt'$, $\ell \in \calS^*$.

To show the lemma, we make a case distinction, first according to whether $p$ is badly cut or not w.r.t.~$\calD$ and $\calL$. 
In the case where it is badly cut, it holds that $\tilde p = \calL(p)$.

\begin{enumerate}
 \item If $\calL(p)$ is badly cut w.r.t.\ $\calD$ and $\opt'$, then $\calL(p) \in \calS^*$: therefore, 
$\tilde p = \calL(p) =  \calS^*(p)$, and $\tilde p$ and $ \calS^*(p)$ are not cut, e.g. $i=-\infty$. Therefore, $s = \calL(p)$ satisfies the lemma, with $\dist(p, \tilde p) + \dist(\tilde p, s) + \eps 2^i = \dist(p, \tilde p)$. 
 \item $\calL(p)$ is not badly cut w.r.t.\ $\calD$ and $\opt'$. Let $s = \calS^*(\calL(p))$, and $i$ be the level at which $\calL(p)$ and $s$ are cut: we write 
 $\dist(p, \tilde p) + \dist(\tilde p, s) + \eps 2^i = \calL_p + \dist(\tilde p, \calS^*) + \eps 2^{i}$. 
 
 By \cref{fact:ineq}, $s \in B(\calL(p), 2\opt'_{\calL(p)})$, and so
  $i \leq i' := \level\lpar \calL(p),  3\opt'_{\calL(p)}\rpar$. 
 Combined with the fact that $\calL(p)$ is not badly cut,  it therefore holds that 
 $\eps 2^{i} \leq \eps 2^{i'} \leq b_3(p)$. 
  \end{enumerate}
 
 We now turn to the case where $p$ is not badly cut. In that case, $\tilde p = p$: we therefore only need to bound $\dist( p, s) + \eps 2^i$. 
 Furthermore, the ball $B(p, 3\calL_p)$ is cut at level $\level(p, 3\calL_p) \leq \log(3\calL_p)+\offset$. 
We make a case distinction according to whether the ball $B(p, 3(\calL_p + \opt'_p))$  is badly cut or not. Note that, by \cref{fact:propSstar}, $ \calS^*(p) \in B(p, 3(\calL_p + \opt'_p))$.
 \begin{enumerate}
  \item If $B(p, 3(\calL_p + \opt'_p))$ is not badly cut, then $p$ and $s = \calS^*(p)$ are cut at level at most $i = \level(p, 3(\calL_p + \opt'_p))$, and $s$ satisfies the lemma: $\dist(p, s) + \eps 2^i \leq \dist(p, \calS^*) + b_2(p)$.
  \item If $B(p, 3(\calL_p + \opt'_p))$ is badly cut, we show that $s= \calS^*(\calL(p))$ satisfies the lemma. In particular,
  \begin{enumerate}
      \item if $\calL(p)$ is badly cut, then $s = \calL(p)$; in that case, since $p$ is not badly cut, $p$ and $s$ are cut at level at most $i = \level(p, 3\calL_p)$ and $\dist(p, s) + \eps 2^i \leq \dist(p, \calS^*) + b_1(p)$.
      \item if $\calL(p)$ is not badly cut, then one can assign $p$ to $s=\calS^*(\calL(p))$. The key observation here is that they are cut at level $i = \max\lpar \level(p, \calL_p), \level(\calL(p),  3\opt'_{\calL(p)}) \rpar$: indeed, from \cref{fact:ineq} $s$ is in the ball $B(\calL(p), 3\opt'_{\calL(p)})$, therefore $s$ and $\calL(p)$ are in the same cell at level $i$; similarly, $p$ and $\calL(p)$ are at the same cell at level $i$, and, therefore, $p$ and $s$ are cut at level at most $i$.  Additionally, from \cref{fact:ineq}, $p$ and $s$ satisfy $\dist(p, s) \leq 3\calL_p + 2\opt'_p$.    
 Therefore,
       \begin{align*}
           \dist(\tilde p, s) + \eps 2^i &\leq 
           3\calL_p + 2\opt'_p + \eps 2^i
       \end{align*}
    Now, since $B(p, 3(\calL_p + \opt'_p))$ is badly cut, $3\calL_p + 2\opt'_p \leq b_2(p)$. Furthermore, by the upper bound on $i$ we have $\eps 2^{i} \leq b_1(p) + b_3(p)$, as both $p$ and $\calL(p)$ are not badly cut. 
    Therefore, $\dist(\tilde p, s) + \eps 2^i \leq b_1(p)+ b_2(p) + b_3(p)$, which concludes the case and the lemma.
  \end{enumerate}
\end{enumerate}
%
\end{proof}

\paragraph{Conclusion of the $k$-median algorithm.} Given this structural result for $k$-median, the very same dynamic program as for $k$-means allows us to compute a $(1+\eps)$-approximate solution.

\section{Other Extensions of the Low Dimensional Algorithm}\label{app:extension}
In this section, we present improved approximation schemes for two clustering variants that were considered in \cite{jacm}, namely prize-collecting clustering and clustering with outliers.

\subsection{Prize-collecting $k$-means and $k$-median}
In the prize-collecting $k$-median and $k$-means, every client $p \in P$ is associated with a penalty $\pi(p)$ that one can pay instead of connecting $p$ to a center. A solution to the problem is a set of $k$ centers $\calS$, with a set of outliers $O \subseteq P$: the cost of that solution is $\sum_{p \in P \setminus O} \cost(p, \calS) + \sum_{p \in O} \pi(p)$. 

The approach of \cite{jacm} is the following: compute a constant-factor approximation $(\calL_S, \calL_O)$, and say a point $p \in \calL_O$ is badly cut if it is badly cut w.r.t.\ $\calD$ and $\opt$ (i.e., the same definition as for centers of $\calL_S$). An outlier is badly cut with probability $\badcut$. The set of badly cut centers is $B_\calD \subseteq \calL_S$, and the set of badly cut outliers is $BO_\calD \subseteq \calL_0$.
Now, an instance has \emph{small distortion} when 
\begin{enumerate}
\item\label{prop:budget-pc} the total budget w.r.t.\ $\calL$ and $\opt'$ is bounded: $\sum_{p \in P} b(p, \eps) \leq \eps (\cost(P, \opt) + \cost(P, \calL))$.
\item\label{prop:costSstar-pc} there exists a solution $(\calS^*, O^*)$ such that $\calS^*$ that contains $B_{\calD}$, $O^*$ contains $BO_\calD$, and
with \\
$\sum_{p \in P \setminus O^*} \lpar \dist(p, \tilde p) + \dist(\tilde p,  \calS^*)\rpar^2 \leq (1+\eps)\cost(P, \opt) + \eps \cost(P, \calL)$,
\item\label{prop:det-pc} for each point $p \in P \setminus O^*$, its detour to $\calS^*$ is at most $b(p, \eps)$, namely there is a center $s \in \calS^*$ cut from $\tilde p$ at level $i$ such that $(\dist(p, \tilde p) + \dist(\tilde p, s) + \eps 2^i)^2 \leq \lpar\dist(p, \tilde p) + \dist(\tilde p, \calS^*)\rpar^2 + b(p, \eps)$.
\end{enumerate}

\cite{jacm} showed how to modify the construction of $\calS^*$ in order to account for outliers: essentially, every time a center of $\opt'$ is closed in Step 1 of the construction, all the clients it serves that are also in $\calL_0$ are added to $O^*$. Additionally, all badly cut outliers are added to $O^*$. Since each of these happens with probability at most $\eps$, the cost of the outliers added this way is at most $\eps \sum_{p \in \calL_O} \pi(p) = O(\eps)\cost(P, \opt)$.

For the other properties, the proof of the structure theorem \cref{thm:smallDist} goes through exactly the same way, showing that with probability at least $2/3$ the instance has small distortion. The dynamic program of \cref{sec:algo} can be easily adapted to solve the prize-collecting version of the problems.

\subsection{$k$-means and $k$-median with Outliers}
In this clustering version, we are given an integer $z$, and the cost function ignores the $z$ most expensive points. More formally, a solution to the problem is a set of $k$ centers $\calS$, with a set of $z$ outliers $O \subseteq P$: the cost of that solution is $\sum_{p \in P \setminus O} \cost(p, \calS)$.

\cite{jacm} showed how to use the ideas from prize-collecting to get a bi-criteria solution, namely a solution with $(1+\eps)z$ outliers and cost $(1+\eps) \cost(P, \opt)$. The idea is that the previous argument for prize-collecting adds new outliers with probability only $\eps$: starting with a set of $z$ outliers, the number of such added outliers is only $O(\eps z)$. The remaining part of the analysis -- bounding detour for the non-outliers -- stays unchanged compared to the standard $k$-median and $k$-means: hence, our improved analysis carries over directly to the outlier case.

\end{document}